\documentclass[11pt]{article}

\usepackage[T1]{fontenc}
\usepackage{iftex}
\ifPDFTeX
  \usepackage[utf8]{inputenc}
\fi
\usepackage{lmodern}
\usepackage[margin=1in]{geometry}
\usepackage{amsmath,amssymb,amsthm,mathtools}
\usepackage{cite}
\usepackage{microtype}
\usepackage{xcolor}

\definecolor{ReferenceBlue}{HTML}{4AA8E8}
\definecolor{EquationPink}{HTML}{F04E98}

\usepackage[colorlinks=true]{hyperref}

\hypersetup{
  citecolor=ReferenceBlue,
  urlcolor=ReferenceBlue,
  linkcolor=black,
  pdftitle={Sharp Plucker Geometry for Three-Copy Werner Distillation},
  pdfauthor={Tianhao Wu},
  pdfsubject={Quantum information, entanglement distillation, partial-trace inequalities},
  pdfkeywords={Werner states, entanglement distillation, partial trace, Plucker coordinates, rank-two operators}
}

\let\OriginalEqref\eqref
\renewcommand{\eqref}[1]{%
  \begingroup
  \hypersetup{linkcolor=EquationPink}%
  \OriginalEqref{#1}%
  \endgroup
}

\newtheorem{theorem}{Theorem}[section]
\newtheorem{proposition}[theorem]{Proposition}
\newtheorem{lemma}[theorem]{Lemma}
\newtheorem{corollary}[theorem]{Corollary}
\theoremstyle{definition}

\theoremstyle{remark}
\newtheorem{remark}[theorem]{Remark}

\newcommand{\Tr}{\operatorname{Tr}}
\newcommand{\rank}{\operatorname{rank}}
\newcommand{\SR}{\operatorname{SR}}
\newcommand{\id}{\operatorname{id}}
\newcommand{\cH}{\mathcal H}
\newcommand{\cL}{\mathcal L}
\newcommand{\cS}{\mathcal S}
\newcommand{\HS}{\mathrm{HS}}
\newcommand{\Real}{\operatorname{Re}}
\newcommand{\ketbra}[2]{|#1\rangle\!\langle#2|}
\newcommand{\Pu}{P_u}
\newcommand{\Pv}{P_v}
\newcommand{\wedgeop}{\mathord{\bigwedge\nolimits^{2}}}

\newcommand{\AuthorOneName}{Tianhao Wu}
\newcommand{\AuthorOneAffiliations}{Department of Physics, University of Illinois, Urbana 61801, USA}
\newcommand{\AuthorOneEmail}{twu49@illinois.edu}
\newcommand{\AuthorTwoName}{Qiran Zou}
\newcommand{\AuthorTwoAffiliations}{National University of Singapore, Singapore}

\title{\textbf{Sharp Pl\"ucker Geometry for Three-Copy Werner
Distillation}}
\author{%
  \AuthorOneName\textsuperscript{1,*}
  \qquad
  \AuthorTwoName\textsuperscript{2}\\[0.75em]
  \small\textsuperscript{1}\AuthorOneAffiliations\\
  \small\textsuperscript{2}\AuthorTwoAffiliations\\[0.45em]
  \small\textsuperscript{*}\emph{Correspondence to \AuthorOneName:}
  \texttt{\AuthorOneEmail}
}

\date{}

\begin{document}

\maketitle

\begin{abstract}
Whether negative-partial-transpose entanglement can remain
undistillable is a longstanding problem in quantum information theory.
We analyze the first unresolved three-copy Werner endpoint using a
sharp dimension-free inequality for complementary partial traces and an
optimal exterior-square inequality for orthonormal tripartite vectors.
The latter identifies local SWAP-parity statistics with metric data of a
decomposable Pl\"ucker bivector.  Together these inequalities prove endpoint
nonnegativity for every positive semidefinite rank-two coefficient
operator and for the complete normal rank-two sector in arbitrary finite
local dimensions.  For genuinely nonnormal operators, an exact
crossed-Gram criterion proves nonnegativity when one local
output--input support overlap is at most two, including every system
with a qubit-sized factor, and when either support plane contains a
product ray.  Explicit anti-state and rank-boundary families establish
optimality of the constants and the rank restriction.  
\end{abstract}

\section{Introduction and main results}

\subsection{Physical setting: Werner states and distillation}
Entanglement distillation takes multiple imperfect bipartite states as input and, using only LOCC, produces fewer output pairs with fidelities approaching that of a maximally entangled state. This conversion is a basic primitive of long-distance quantum communication and gives the operational distinction between free and bound entanglement\cite{BennettEtAl,Horodecki1997,Horodecki1998,DurEtAl,LewensteinPrimer}.
A bipartite state \(\rho\) is said to have positive partial transpose (PPT) if \(\rho^\Gamma\succeq0\), where \(\Gamma\) transposes one party in a fixed product basis; positivity of \(\rho^\Gamma\) is independent of the chosen local basis \cite{Peres,HorodeckiPPT}.  Every PPT state is undistillable.  Whether
every negative-partial-transpose (NPT) state is distillable is not known;
an NPT counterexample would constitute NPT bound entanglement \cite{DiVincenzoEtAl,WatrousCopies,Clarisse,ViannaDoherty,
PankowskiEtAl,Djokovic,HorodeckiProblems}.
Such a counterexample would also have consequences for additivity of
distillable entanglement and is closely connected with tensor-stable
positive maps \cite{ShorSmolinTerhal,MullerHermesReebWolf}.
Comparison with the larger class of PPT-preserving operations isolates
restrictions intrinsic to LOCC beyond partial-transpose negativity
\cite{EggelingEtAl}.
Both single-photon-pair purification assisted by hyperentanglement and
the activation of PPT bound entanglement have been realized
experimentally \cite{EckerEtAl,KanedaEtAl}; these experiments establish the operational
reality of the resource distinction, although neither addresses the NPT
finite-copy conjecture studied here.

Werner states form the canonical symmetric test family
\cite{Werner,VollbrechtWerner}
\begin{equation}
 \rho_\alpha^{(d)}
 =\frac{I_{d^2}+\alpha F_d}{d^2+\alpha d},
 \qquad -1\leq\alpha\leq1,
 \label{eq:werner}
\end{equation}
on \(\mathbb C^d\otimes\mathbb C^d\), where
\(F_d(x\otimes y)=y\otimes x\).  These states are invariant under \(U\otimes U\) conjugation.  They are separable precisely when \(\alpha\geq-1/d\) and one-copy distillable precisely when
\(\alpha<-1/2\) \cite{Werner,DurEtAl,VollbrechtWerner}.  Thus, for \(d\geq3\), the interval \(-1/2\leq\alpha<-1/d\) consists of NPT states that are not one-copy distillable.  The recovery of operational quantum resources from two-qutrit Werner states by local filtering has also been demonstrated experimentally \cite{FangEtAl}.

A state \(\rho\) is \(k\)-copy distillable if an LOCC protocol can map \(\rho^{\otimes k}\), with nonzero probability, to an entangled two-qubit state.  The Schmidt rank of a bipartite vector is the number of nonzero terms in its Schmidt decomposition.  Equivalently, there is
a unit vector \(\psi\) of Schmidt rank at most two such that
\[
 \langle\psi,(\rho^\Gamma)^{\otimes k}\psi\rangle<0
\]
\cite{Horodecki1997,Horodecki1998,DurEtAl}.  The finite-copy criterion
for \eqref{eq:werner} is most transparent after vectorization.
Write
\[
 |\operatorname{vec}C\rangle
 =\sum_{\boldsymbol a,\boldsymbol b}
 C_{\boldsymbol a,\boldsymbol b}
 |\boldsymbol a\rangle_A|\boldsymbol b\rangle_B .
\]
The Schmidt rank of this vector across \(A:B\) equals \(\rank C\).
Because \(F_d^\Gamma=\ketbra{\Omega_d}{\Omega_d}\), with
\(|\Omega_d\rangle=\sum_j|jj\rangle\), direct contraction gives
\cite{QianEtAl,CostaRico,QietAl,BhartiGajjalaHaug}
\begin{equation}
 \langle\operatorname{vec}C|
 [(\rho_\alpha^{(d)})^\Gamma]^{\otimes k}
 |\operatorname{vec}C\rangle
 =\frac{q_k(\alpha,C)}{(d^2+\alpha d)^k},
 \qquad
 q_k(\alpha,C)
 =\sum_{J\subseteq[k]}\alpha^{|J|}
   \|\Tr_JC\|_2^2 .
 \label{eq:qk-criterion}
\end{equation}
Here and below, \([k]=\{1,\ldots,k\}\).
Consequently, \(\rho_\alpha^{(d)}\) is \(k\)-copy distillable if and
only if \(q_k(\alpha,C)<0\) for some \(C\) of rank at most two.  This
criterion is recalled with its index contraction in
Appendix~\ref{app:werner-criterion}.

Four recent works have established the exact two-copy threshold \cite{FuGaoPark,SongChen,FraserEtAl,BhartiGajjalaHaug}.  The three-copy endpoint \(\alpha=-1/2\) is the first unresolved layer beyond that threshold.  For an operator \(C\) on
\(\cH_1\otimes\cH_2\otimes\cH_3\), it is governed by
\begin{equation}
 q_3(C)=\|C\|_2^2-\frac12\sum_{i=1}^3\|\Tr_iC\|_2^2
 +\frac14\sum_{1\leq i<j\leq3}\|\Tr_{ij}C\|_2^2
 -\frac18|\Tr C|^2 .
 \label{eq:q3}
\end{equation}
The tensor factors in \(\cH_1\otimes\cH_2\otimes\cH_3\) label the three resource copies; the physical bipartition remains \(A_1A_2A_3:B_1B_2B_3\).  Thus ``tripartite'' below describes the
coefficient-space tensor structure, not a third communicating party.  The alternating weights in \eqref{eq:q3} form a signed ledger of how much Hilbert--Schmidt coherence remains after zero, one, two, or all three copy slots are contracted.  The two-copy theorem cannot simply be reapplied after tracing out one slot, because a partial trace of a rank-two operator need not remain rank two.  At three copies the reductions must instead be controlled simultaneously.  Estimating the cuts independently forgets that they all originate from the same two-dimensional global support; the Pl\"ucker step restores precisely this common-origin constraint.

We use the following sesquilinear polarization throughout:
\begin{equation}
 b_3(X,Y)=
 \sum_{J\subseteq\{1,2,3\}}
 \left(-\frac12\right)^{|J|}
 \Tr\!\left[(\Tr_JX)^\dagger\Tr_JY\right],
 \qquad q_3(X)=b_3(X,X).
 \label{eq:b3}
\end{equation}
Here \(\Tr_J\) traces out the tensor factors indexed by \(J\), with \(\Tr_\varnothing C=C\).  The local spaces need not have equal dimensions in any of the results below.  The unrestricted endpoint asks whether \eqref{eq:q3} is nonnegative on every \(C\) of rank at most two. Here \(C\) is the coefficient matrix of a distillation test vector and need not be a density operator.  We first prove nonnegativity on the positive semidefinite cone.  Combining the two spectral signs then closes the complete normal rank-two sector.  Polarization beyond normality isolates one exact crossed-Gram determinant.  We prove two broad closures of that determinant: one local output--input support overlap is at most two, or either global support plane contains a product ray.  The normal result also gives the common initial--final two-plane case directly.  Appendix~\ref{app:additional-sectors} treats
the more specialized local-traceless and crossed product-code sectors.

Throughout, all Hilbert spaces are finite-dimensional and complex.
For a Hilbert space \(\mathcal H\), \(\mathcal L(\mathcal H)\) denotes its linear operators and \(I\) the identity operator.  The symbols \({}^\dagger\), \({}^T\), and an overline denote, respectively, the Hilbert-space adjoint, matrix transpose, and entrywise complex
conjugation in the fixed product bases; \(\Real z\) is the real part of a scalar \(z\).  The orders \(A\succeq B\) and \(A\preceq B\) are the Loewner orders, and \(\operatorname{ran}A\) is the range of \(A\).
The identity superoperator is denoted by \(\id\), and \(\SR(\psi)\) denotes the Schmidt rank of a bipartite vector.

The inner product \(\langle\cdot,\cdot\rangle\) is conjugate-linear in its first argument.  All partial traces are the standard unnormalized partial traces.  For \(J\subseteq[k]\), \(\bar J=J^c=[k]\setminus J\). In the tripartite setting, \(\bar i=\{1,2,3\}\setminus\{i\}\),
\(d_i=\dim\mathcal H_i\), and
\[
 C_i:=\Tr_{\bar i}C,\qquad C_{\bar i}:=\Tr_iC.
\]
Thus the subscript of \(\Tr_J\) records the factors that are traced out, whereas the subscript of a reduced operator records the factors that remain.  Finally, \(|X|=(X^\dagger X)^{1/2}\), and
\[
 \|X\|_p=(\Tr|X|^p)^{1/p}\quad(1\leq p<\infty)
\]
denotes the Schatten norm.  In particular, \(\|\cdot\|_1\) is the trace or nuclear norm, \(\|\cdot\|_2\) is the Hilbert--Schmidt norm, \(\|\cdot\|_\infty\) is the operator norm, and
\(\langle A,B\rangle_{\HS}:=\Tr(A^\dagger B)\).

\subsection{Three sharp inequalities}

We present three theorems as the main results in this paper.  The first one compresses each bipartite cut without a dimension or rank penalty.  Then the second one restores the common two-plane compatibility across the three cuts, and the third statement converts that geometry into the required distillation inequality.  Their composition is the three-copy mechanism.

Partial-trace norm inequalities supply the dimension-free input to our argument.  Their development has a substantial history. Audenaert proved a Schatten-\(q\) inequality for positive operators, Rastegin developed bounds for broader classes of unitarily invariant
norms, and Costa Rico and Wolf recently extended the subadditivity
framework to nonnormal matrices \cite{Audenaert,Rastegin,CostaRicoWolf}.  Our estimate removes both positivity and rank restrictions by measuring the original operator in nuclear norm, while retaining a dimension-free joint quadratic bound. This rank-unrestricted estimate is unprecedented in the existing partial-trace literature. More precisely, Proposition~11 of \cite{CostaRicoWolf} gives, for rank-one \(M\),
\[
 \|\Tr_{\mathcal K}M\|_2^2+\|\Tr_{\mathcal L}M\|_2^2
 \leq \|M\|_2^2+|\Tr M|^2,
\]
whereas \cite{FraserEtAl} 
prove, for \(\rank M\leq r\),
\[
 \|\Tr_{\mathcal K}M\|_2^2+\|\Tr_{\mathcal L}M\|_2^2
 \leq r\|M\|_2^2+r^{-1}|\Tr M|^2 .
\]  The following estimate complements these rank-constrained Hilbert--Schmidt bounds by replacing the rank parameter with the nuclear norm.

\begin{theorem}[Sharp quadratic partial-trace inequality]
\label{thm:nuclear-pt}
For every \(X\in\cL(\mathcal K\otimes\mathcal L)\),
\begin{equation}
 \|\Tr_{\mathcal K}X\|_2^2+\|\Tr_{\mathcal L}X\|_2^2
 \leq\|X\|_1^2+|\Tr X|^2 .
 \label{eq:nuclear-pt}
\end{equation}
Neither coefficient on the right can be decreased while the other is
held fixed.
\end{theorem}

Nuclear norms also enter entanglement theory through projective cross-norm and realignment criteria \cite{Rudolph,ChenWu}. 
The common feature of realignment and Theorem~\ref{thm:nuclear-pt} is that \(\|\cdot\|_1\) retains the full singular-value aggregate that a Hilbert--Schmidt estimate may discard.

The geometric consequence acts on an orthonormal pair \(u,v\) in 
\(\cH_1\otimes\cH_2\otimes\cH_3\).  On two copies, put \(P_i^\pm=(I\pm F_i)/2\), where 
\(F_i\) swaps the two copies of \(\cH_i\).  These three swaps commute.  If \(S\subseteq\{1,2,3\}\), let
\[
 \Pi_S=\prod_{i\in S}P_i^-\prod_{i\notin S}P_i^+,
\]
and define
\begin{align}
 a(u)&=\sum_{|S|=2}\|\Pi_S(u\otimes u)\|^2,\\
 p_2(u,v)&=\sum_{|S|=2}\|\Pi_S(u\otimes v)\|^2,\\
 p_3(u,v)&=\|P_1^-P_2^-P_3^-(u\otimes v)\|^2.
 \label{eq:ap-def}
\end{align}
The two copies in these definitions are auxiliary replicas of the coefficient-space rays.
Thus \(\|\Pi_S(u\otimes v)\|^2\) is the joint probability of obtaining
antisymmetric local SWAP parity precisely at the sites in \(S\).
Accordingly, \(a(u)\) is the total two-minus probability for two copies of \(u\), while \(p_2(u,v)\) and \(p_3(u,v)\) are the two-minus and three-minus probabilities for the ordered pair \(u\otimes v\). Symmetric and antisymmetric two-copy projectors also underlie
multipartite-concurrence constructions and scalable parity-based entanglement estimation \cite{MintertEtAl,AolitaEtAl}.  Here their joint three-site distribution is constrained by one decomposable Pl\"ucker bivector. Such parity probabilities, and the subsystem purities to which they reduce, can be measured by two-copy interference
\cite{EkertEtAl,AlvesEtAl,BovinoEtAl,IslamEtAl,ElbenEtAl,MillerEtAl}.
For qubits, \(P_i^-\) and \(P_i^+\) are respectively the singlet and triplet projectors.  Recent trapped-ion experiments identify the corresponding Bell-sampling triplet statistics with quantum shadow enumerators \cite{MillerEtAl}.

\begin{theorem}[Sharp exterior-square (Pl\"ucker-coordinate) inequality]
\label{thm:sharp-plucker}
Every orthonormal tripartite pair satisfies
\begin{equation}
 3p_3(u,v)-p_2(u,v)\leq\sqrt{a(u)a(v)} .
 \label{eq:sharp-plucker}
\end{equation}
The constant is optimal in every set of local dimensions containing a
\(2\times2\times2\) subspace.
\end{theorem}

The term ``Pl\"ucker'' records that \(u\wedge v\) is a decomposable bivector, hence a point of the Grassmannian in its Pl\"ucker embedding \cite{Levay,Harris,Landsberg}.  Its fully locally antisymmetric component controls \(p_3\), while \(a(u),a(v)\) are sums of squared
second-compound norms.  Theorem~\ref{thm:sharp-plucker} retains the decomposability constraint that is lost if the local swap-sector weights are treated as arbitrary probabilities.  The theorem uses the metric geometry of the Pl\"ucker embedding. Each cut sees a different projection of the same oriented two-area \(u\wedge v\).

\begin{theorem}[Positive rank-two consequence]
\label{thm:strong}
If \(C\succeq0\) and \(\rank C\leq2\), then
\begin{align}
 \cS(C)&\geq0,\label{eq:S-positive}\\
 q_3(C)&\geq
 \frac{2\Tr C^2-(\Tr C)^2}{8}\geq0,
 \label{eq:strong-bound}
\end{align}
where
\begin{equation}
 \cS(C)=3\|C\|_2^2+\sum_i\|\Tr_{\bar i}C\|_2^2
       -2\sum_i\|\Tr_iC\|_2^2 .
 \label{eq:S-def}
\end{equation}
Here \(\Tr_{\bar i}\) leaves subsystem \(i\), whereas \(\Tr_i\) leaves the complementary pair.
\end{theorem}

For a normalized state,
\(\cS(C)\geq0\) is equivalently
\begin{equation}
 3S_L(C)+\sum_iS_L(C_i)
 \leq2\sum_iS_L(C_{\bar i}),\qquad
 S_L(\rho)=1-\Tr\rho^2 .
 \label{eq:sssa}
\end{equation}
The linear entropy is a directly measurable second-order impurity, distinct from the von Neumann entropy.  Equation~\eqref{eq:sssa} is therefore a purity constraint, separate from the usual strong subadditivity theorem.
It belongs to the family of quantum shadow and sector-length inequalities
\cite{ShorLaflamme,RainsShadow,Rains,EltschkaEtAl,WyderkaGuhne}.
Wyderka and G\"uhne proved \eqref{eq:sssa} for all three-qubit states
\cite{WyderkaGuhne}; Theorem~\ref{thm:strong} allows arbitrary local dimensions under the rank-two hypothesis.

The proof begins with the exact two-vector reduction in Section~\ref{sec:reduction} and the sharp projection lemma in
Section~\ref{sec:projection}.  Sections~\ref{sec:nuclear-proof} and
\ref{sec:plucker-proof} establish the two main inequalities.
Section~\ref{sec:consequences} derives their positive, normal, and common-plane consequences.  Section~\ref{sec:nonnormal} identifies the remaining nonnormal determinant and proves two broad closures.
Section~\ref{sec:sharpness} constructs equality families and rank
boundaries; Section~\ref{sec:physical-meaning} gives the operational
interpretation.  The appendices contain the contraction identities and
two additional closed sectors.

\section{Exact two-vector reduction}
\label{sec:reduction}

We first separate the two spectral rays and express every mixed term in
local SWAP-parity variables.  Let
\begin{equation}
 C=\lambda\Pu+\mu\Pv,\qquad
 \lambda,\mu\geq0,\qquad \langle u,v\rangle=0,
 \label{eq:spectral}
\end{equation}
where \(P_w=\ketbra{w}{w}\).  For a unit tripartite vector \(w\), write
\(\rho_S^w=\Tr_{\bar S}P_w\), and set
\begin{equation}
 s_w=\sum_{i=1}^3\Tr[(\rho_i^w)^2],
 \qquad
 e_i(w)=1-\Tr[(\rho_i^w)^2].
 \label{eq:se-def}
\end{equation}
Thus \(e_i(w)\) is the linear entropy of entanglement across the pure
bipartition \(i:\bar i\).
Expanding the local swap projectors gives
\begin{equation}
 a(w)=\frac14(3-s_w)=\frac14\sum_i e_i(w).
 \label{eq:a-purity}
\end{equation}

For an orthonormal pair \(u,v\), define the marginal overlaps
\[
 x_i=\Tr(\rho_i^u\rho_i^v),\qquad
 y_{ij}=\Tr(\rho_{ij}^u\rho_{ij}^v).
\]
The Walsh expansion is the discrete Fourier transform on
\(\mathbb Z_2^3\).  Expanding the three commuting parity projectors and
using the SWAP identity gives, for every \(S\subseteq\{1,2,3\}\),
\begin{equation}
 \|\Pi_S(u\otimes v)\|^2
 =\frac1{8}\sum_{T\subseteq\{1,2,3\}}
 (-1)^{|S\cap T|}\Tr(\rho_T^u\rho_T^v).
 \label{eq:walsh-sector}
\end{equation}
The empty-set character is one, while the full-set character is
\(\Tr(P_uP_v)=|\langle u,v\rangle|^2=0\).  Summing the characters for
the all-minus outcome and for the three outcomes with exactly two minus
signs therefore gives
\begin{align}
 p_3&=\frac18\left(1-\sum_i x_i+\sum_{i<j}y_{ij}\right),\\
 p_2&=\frac18\left(3-\sum_i x_i-\sum_{i<j}y_{ij}\right).
 \label{eq:p2p3}
\end{align}
In particular,
\begin{equation}
 4(p_2-3p_3)=\sum_i x_i-2\sum_{i<j}y_{ij}.
 \label{eq:p-cross}
\end{equation}

The following algebraic identity does not require positivity, normality,
or a rank assumption.  For every tripartite operator \(C\), direct
coefficient comparison between \eqref{eq:q3} and \eqref{eq:S-def}
gives
\begin{equation}
 q_3(C)
 =\frac14\cS(C)
 +\frac18\bigl(2\|C\|_2^2-|\Tr C|^2\bigr).
 \label{eq:q-S}
\end{equation}

\begin{proposition}[Exact positive spectral bridge]
\label{prop:spectral-bridge}
For \(C=\lambda\Pu+\mu\Pv\) as above,
\begin{equation}
 \cS(C)
 =4\left[
 \lambda^2a(u)+\mu^2a(v)
 +2\lambda\mu\bigl(p_2(u,v)-3p_3(u,v)\bigr)
 \right].
 \label{eq:S-exterior}
\end{equation}
\end{proposition}

\begin{proof}
For a pure state, complementary purities agree.  Hence
\[
 \cS(P_w)=3-\sum_i\Tr[(\rho_i^w)^2]=4a(w).
\]
The cross coefficient in \(\cS(\lambda\Pu+\mu\Pv)\) is
\[
 2\lambda\mu
 \left(\sum_i x_i-2\sum_{i<j}y_{ij}\right),
\]
because \(\Tr(\Pu\Pv)=0\).  Equation~\eqref{eq:p-cross} proves
\eqref{eq:S-exterior}.
\end{proof}

\begin{remark}
Theorem~\ref{thm:sharp-plucker} immediately turns
\eqref{eq:S-exterior} into a perfect square:
\[
 \frac14\cS(C)\geq
 \bigl(\lambda\sqrt{a(u)}-\mu\sqrt{a(v)}\bigr)^2.
\]
The remaining sections establish the Pl\"ucker estimate without
assuming local dimensions.
\end{remark}

\section{The double-antisymmetric projection}
\label{sec:projection}

The sharp partial-trace constant comes from one projection problem.
Let \(\mathcal K,\mathcal L\) be finite-dimensional Hilbert spaces.  On
\((\mathcal K\otimes\mathcal L)^{\otimes2}\), let
\(F_{\mathcal K}\) and \(F_{\mathcal L}\) be the local copy swaps and put
\[
 Q=P_{\mathcal K}^-P_{\mathcal L}^-,
 \qquad
 P_{\mathcal K}^-=\frac{I-F_{\mathcal K}}2,\quad
 P_{\mathcal L}^-=\frac{I-F_{\mathcal L}}2.
\]
For an orthogonal projection \(P\), the variational quantity
\[
 \sup_{\substack{\|\psi\|=1\\\SR(\psi)\leq k}}
 \langle\psi,P\psi\rangle
\]
is its \(S(k)\)-operator norm \cite{JohnstonKribs,JohnstonEtAl}.
The following lemma determines the \(S(2)\)-norm of \(Q\).  It is the
sharp form of the half-property studied by Pankowski \emph{et al.} and
is the geometric input of the recent two-copy solution
\cite{PankowskiEtAl,FuGaoPark}.  Fu, Gao, and Park state their
Theorem~2.4 for equal local dimensions; the proof below gives the
rectangular \(\mathcal K\otimes\mathcal L\) form needed here.

\begin{lemma}[Sharp double-antisymmetric projection]
\label{lem:double-antisymmetric}
If \(\psi\) is a unit vector of Schmidt rank at most two across the copy
cut
\[
 (\mathcal K\otimes\mathcal L)_{\mathrm{copy}\,1}:
 (\mathcal K\otimes\mathcal L)_{\mathrm{copy}\,2},
\]
then
\begin{equation}
 \langle\psi,Q\psi\rangle\leq\frac12.
 \label{eq:double-antisymmetric}
\end{equation}
If \(\dim\mathcal K,\dim\mathcal L\geq2\), the constant \(1/2\) is optimal.
\end{lemma}

\begin{proof}
This is the sharp projection statement used by Fu, Gao, and Park \cite[Theorem 2.4]{FuGaoPark}; we include a self-contained proof. If \(\dim\mathcal K<2\) or \(\dim\mathcal L<2\), then one local antisymmetric projector vanishes, so \(Q=0\) and the assertion is
immediate.  Hence assume \(\dim\mathcal K,\dim\mathcal L\geq2\).
For a unit \(\phi\in\operatorname{ran}Q\), the full copy swap
\(F_{\mathcal K}F_{\mathcal L}\) fixes \(\phi\).  Takagi factorization
therefore gives
\[
 \phi=\sum_j s_j e_j\otimes e_j,\qquad
 s_1\geq s_2\geq\cdots\geq0,\qquad \sum_js_j^2=1.
\]
The largest squared overlap of \(\phi\) with a unit vector of Schmidt
rank at most two is \(s_1^2+s_2^2\).  Because \(Q\) is an orthogonal
projection, taking the two suprema in either order gives
\begin{equation}
\begin{split}
 \sup_{\substack{\|\psi\|=1\\\SR(\psi)\leq2}}
 \langle\psi,Q\psi\rangle
 &=
 \sup_{\substack{\|\psi\|=1\\\SR(\psi)\leq2}}
 \sup_{\substack{\phi\in\operatorname{ran}Q\\\|\phi\|=1}}
 |\langle\phi,\psi\rangle|^2\\
 &=
 \sup_{\substack{\phi\in\operatorname{ran}Q\\\|\phi\|=1}}
 \bigl(s_1(\phi)^2+s_2(\phi)^2\bigr).
\end{split}
 \label{eq:S2-proj}
\end{equation}

Fix \(\phi\), put \(t=s_1^2+s_2^2\), and define
\[
 \eta=t^{-1/2}
 (s_1e_1\otimes e_1+s_2e_2\otimes e_2)
 =c_1e_1\otimes e_1+c_2e_2\otimes e_2.
\]
Identify \(e_j\in\mathcal K\otimes\mathcal L\) with a rectangular matrix \(E_j\).  Since \(\eta\) is symmetric under the full copy swap,
\[
 \langle\eta,Q\eta\rangle
 =\frac12\bigl(1-\langle\eta,F_{\mathcal L}\eta\rangle\bigr).
\]
The remaining expectation is
\begin{align*}
 \langle\eta,F_{\mathcal L}\eta\rangle
 &=c_1^2\Tr[(E_1^\dagger E_1)^2]
   +c_2^2\Tr[(E_2^\dagger E_2)^2] 
   +2c_1c_2\Real\Tr(E_1^\dagger E_2E_1^\dagger E_2)\\
 &\geq
 \bigl(c_1\|E_1\|_4^2-c_2\|E_2\|_4^2\bigr)^2\geq0.
\end{align*}
Indeed,
\[
 |\Tr(E_1^\dagger E_2E_1^\dagger E_2)|
 \leq\|E_1^\dagger E_2\|_2^2
 \leq\|E_1\|_4^2\|E_2\|_4^2
\]
by Schatten H\"older.  Thus
\(\|Q\eta\|^2=\langle\eta,Q\eta\rangle\leq1/2\).  Since
\(Q\phi=\phi\),
\[
 t=|\langle\eta,\phi\rangle|^2
   =|\langle Q\eta,\phi\rangle|^2
   \leq\|Q\eta\|^2\leq\frac12.
\]
Taking the supremum in \eqref{eq:S2-proj} proves the bound. For optimality, choose orthonormal local vectors
\(|1\rangle,|2\rangle\) and let
\[
 \psi=\frac{
 |11\rangle_{\mathrm{copy}\,1}|22\rangle_{\mathrm{copy}\,2}
 +|22\rangle_{\mathrm{copy}\,1}|11\rangle_{\mathrm{copy}\,2}}
 {\sqrt2}.
\]
It has Schmidt rank two across the copy cut and
\(\langle\psi,Q\psi\rangle=1/2\).
\end{proof}

\section{Proof of the quadratic partial-trace inequality}
\label{sec:nuclear-proof}

The SVD transfers the projection bound from two-copy vectors to an
arbitrary bipartite operator.

\begin{proof}[Proof of Theorem~\ref{thm:nuclear-pt}]
Take a singular-value decomposition
\[
 X=\sum_{\alpha=1}^r s_\alpha
 \ketbra{u_\alpha}{v_\alpha},\qquad s_\alpha>0,
\]
where each of the two vector families is orthonormal.  Put
\[
 X_\alpha=\ketbra{u_\alpha}{v_\alpha},\qquad
 \Phi(Z)=(\Tr_{\mathcal K}Z,\Tr_{\mathcal L}Z).
\]
The two-component target of \(\Phi\) carries the direct-sum
Hilbert--Schmidt inner product.
With \(F=F_{\mathcal K}F_{\mathcal L}\), the operator
\[
 K=F_{\mathcal K}+F_{\mathcal L}-I=F-4Q
\]
satisfies \(QF=FQ=Q\).  The swap trick gives
\begin{align}
 &\langle\Phi(\ketbra{u}{v}),\Phi(\ketbra{s}{t})\rangle_{\HS}
 -\overline{\Tr(\ketbra{u}{v})}\Tr(\ketbra{s}{t})\notag\\
 &\hspace{32mm}
 =\langle u\otimes t,K(v\otimes s)\rangle .
 \label{eq:swap-kernel}
\end{align}
For example, the \(\Tr_{\mathcal K}\) term on the left equals
\(\langle u\otimes t,F_{\mathcal L}(v\otimes s)\rangle\).

Define
\[
 \gamma_{\alpha\beta}
 =\langle\Phi(X_\alpha),\Phi(X_\beta)\rangle_{\HS}
  -\overline{\Tr X_\alpha}\Tr X_\beta .
\]
On the diagonal, \eqref{eq:swap-kernel} gives
\[
 \gamma_{\alpha\alpha}
 =1-4\|Q(u_\alpha\otimes v_\alpha)\|^2\leq1.
\]
For \(\alpha\neq\beta\), the full-swap contribution vanishes by
orthogonality.  Set
\[
 a=u_\alpha\otimes v_\beta,\qquad
 b=u_\beta\otimes v_\alpha,\qquad
 z=\langle a,Qb\rangle .
\]
Equivalently,
\[
 \gamma_{\alpha\beta}
 =\langle a,(I-4Q)b\rangle
 =-4\langle a,Qb\rangle=-4z,
\]
because \(\langle a,b\rangle=0\).  The vectors \(a,b\) are orthonormal,
and
\(\eta_\theta=(a+e^{i\theta}b)/\sqrt2\) has Schmidt rank at most two
across the copy cut.  Choose \(\theta\) so that the cross term is
\(|z|\), and set
\[
 p=\langle a,Qa\rangle,\qquad q=\langle b,Qb\rangle.
\]
Lemma~\ref{lem:double-antisymmetric} and Cauchy--Schwarz in the
semidefinite inner product induced by \(Q\) give, respectively,
\[
 p+q+2|z|\leq1,\qquad
 2|z|\leq2\sqrt{pq}\leq p+q.
\]
Consequently \(|z|\leq1/4\), and hence
\(\Real\gamma_{\alpha\beta}\leq1\) also off the diagonal.
It follows that
\begin{align*}
 &\|\Tr_{\mathcal K}X\|_2^2+\|\Tr_{\mathcal L}X\|_2^2-|\Tr X|^2\\
 &=\sum_{\alpha,\beta}s_\alpha s_\beta
   \Real\gamma_{\alpha\beta}
 \leq\left(\sum_\alpha s_\alpha\right)^2
 =\|X\|_1^2.
\end{align*}
The double sum is real because
\(\gamma_{\beta\alpha}=\overline{\gamma_{\alpha\beta}}\).

For a rank-one projector onto a product of unit vectors, the two sides
of \eqref{eq:nuclear-pt} equal two; with the coefficient of
\(\|X\|_1^2\) fixed, this forces the coefficient of \(|\Tr X|^2\) to be
at least one.  Conversely, take
\(X=\ketbra{a}{a}\otimes\ketbra{b}{b'}\) with
\(a,b,b'\) unit and \(\langle b,b'\rangle=0\).  Its left side and
\(\|X\|_1^2\) both equal
one, while \(\Tr X=0\), forcing the coefficient of
\(\|X\|_1^2\) to be at least one.
\end{proof}

\section{Exterior-square control and the sharp Pl\"ucker theorem}
\label{sec:plucker-proof}

We now retain the second compound of each reduced transition operator.
Use the standard Hilbert exterior-power normalization
\[
 e_i\wedge e_j\ \longmapsto\
 \frac{e_i\otimes e_j-e_j\otimes e_i}{\sqrt2}
 \qquad(i<j).
\]
Thus the wedges of an orthonormal basis are orthonormal.
The symbol \(\wedgeop X\) denotes the induced map on the second
exterior power; if
\(\sigma_1(X)\geq\sigma_2(X)\geq\cdots\) are the singular values of
\(X\), then the singular values of \(\wedgeop X\) are
\(\sigma_i(X)\sigma_j(X)\), \(i<j\).

\begin{lemma}[Transition exterior-square bound]
\label{lem:wedge-bound}
Let \(u,v\) be normalized vectors in \(\cH_{\bar k}\otimes\cH_k\), and put
\[
 X_k=\Tr_k\ketbra{u}{v}.
\]
Then
\begin{equation}
 \|X_k\|_1^2-\|X_k\|_2^2
 \leq
 \sqrt{e_k(u)e_k(v)} .
 \label{eq:wedge-bound}
\end{equation}
\end{lemma}

\begin{proof}
Flatten \(u,v\) across \(\bar k:k\) into rectangular matrices
\(U,V:\cH_k\to\cH_{\bar k}\).  Then \(X_k=UV^\dagger\) and
\(\|U\|_2=\|V\|_2=1\).  With the stated vectorization convention,
\[
 U^\dagger U=(\rho_k^u)^T,\qquad
 V^\dagger V=(\rho_k^v)^T;
\]
the transpose does not change either spectrum or purity.  The singular values of
\(\wedgeop X_k\) are the pairwise products of those of \(X_k\), so
\[
 \|X_k\|_1^2-\|X_k\|_2^2=2\|\wedgeop X_k\|_1.
\]
These standard compound-matrix identities hold for rectangular maps as
well \cite{HornJohnson}.
Functoriality of exterior powers and Schatten H\"older imply
\[
 \|\wedgeop X_k\|_1
 \leq\|\wedgeop U\|_2\|\wedgeop V\|_2.
\]
Finally,
\[
 \|\wedgeop U\|_2^2
 =\frac12\left(\|U\|_2^4-\|U^\dagger U\|_2^2\right)
 =\frac12\left(1-\Tr[(\rho_k^u)^2]\right)
 =\frac12e_k(u),
\]
and similarly for \(V\).  Combining these identities proves
\eqref{eq:wedge-bound}.
\end{proof}

\begin{proof}[Proof of Theorem~\ref{thm:sharp-plucker}]
For each \(k\in\{1,2,3\}\), let
\(\{i,j,k\}=\{1,2,3\}\), and consider
\[
 X_k=\Tr_k\ketbra{u}{v}\in\cL(\cH_i\otimes\cH_j).
\]
The transition-operator identity
\begin{equation}
 \|\Tr_S\ketbra{u}{v}\|_2^2
 =\Tr(\rho_S^u\rho_S^v)
 \label{eq:transition-identity}
\end{equation}
and \(u\perp v\) give
\[
 \|X_k\|_2^2=x_k,\qquad
 \|\Tr_iX_k\|_2^2=y_{ik},\qquad
 \|\Tr_jX_k\|_2^2=y_{jk},\qquad
 \Tr X_k=0.
\]
Theorem~\ref{thm:nuclear-pt} and
Lemma~\ref{lem:wedge-bound} therefore yield the cutwise estimate
\begin{equation}
 y_{ik}+y_{jk}-x_k
 \leq\sqrt{e_k(u)e_k(v)}.
 \label{eq:one-k-bound}
\end{equation}
Summing over \(k\) counts every \(y_{ij}\) twice.  Cauchy--Schwarz then
gives
\[
 2\sum_{i<j}y_{ij}-\sum_kx_k
 \leq\sum_k\sqrt{e_k(u)e_k(v)}
 \leq\sqrt{\sum_ke_k(u)\sum_ke_k(v)}.
\]
Using \eqref{eq:p2p3} and \eqref{eq:a-purity},
\[
 3p_3-p_2
 =\frac14\left(2\sum_{i<j}y_{ij}-\sum_kx_k\right)
 \leq\sqrt{a(u)a(v)}.
\]
Proposition~\ref{prop:antistate} below supplies equality examples.
\end{proof}

\begin{remark}[ ]
The nuclear norm in Theorem~\ref{thm:nuclear-pt} is not replaced by a rank bound.  Instead, Lemma~\ref{lem:wedge-bound} keeps the second compound matrix of \(X_k=U_kV_k^\dagger\).  Only after obtaining the three cutwise geometric means in \eqref{eq:one-k-bound} is the ordinary
Cauchy--Schwarz inequality applied.
\end{remark}

\section{Rank-two consequences}
\label{sec:consequences}

The Pl\"ucker estimate closes the positive spectral sector through the exact bridge \eqref{eq:S-exterior}.

\begin{proof}[Proof of Theorem~\ref{thm:strong}]
Equations~\eqref{eq:S-exterior} and
\eqref{eq:sharp-plucker} give the stronger square
\begin{equation}
 \frac14\cS(C)
 \geq
 \bigl(\lambda\sqrt{a(u)}-\mu\sqrt{a(v)}\bigr)^2
 \geq0.
 \label{eq:sector-square}
\end{equation}
Moreover,
\[
 \|C\|_2^2=\lambda^2+\mu^2,\qquad
 \Tr C=\lambda+\mu,
\]
so
\[
 2\|C\|_2^2-|\Tr C|^2=(\lambda-\mu)^2\geq0.
\]
Substitution into \eqref{eq:q-S} proves
\eqref{eq:strong-bound}.
\end{proof}

Recall that \(C\) is normal when \(C^\dagger C=CC^\dagger\),
equivalently when it admits an orthonormal spectral basis.
Theorem~\ref{thm:strong} controls the same-phase positive spectral
combination.  The opposite-sign theorem in
\cite[Proposition 3]{CostaRico} controls the other real extreme of the
polarized cross term; together, the two bounds cover arbitrary relative
phases.  In the present normalization, the opposite-sign result states
\[
 q_3(\alpha P_u-\beta P_v)\geq0
 \qquad(u\perp v,\ \alpha,\beta>0).
\]

\begin{corollary}[Normal rank-two endpoint]
\label{cor:normal-rank2}
Every normal tripartite operator \(C\) of rank at most two obeys
\(q_3(C)\geq0\).
\end{corollary}

\begin{proof}
Write
\[
 C=z_1\Pu+z_2\Pv,\qquad u\perp v,
\]
and use the sesquilinear polarization \(b_3\) of \(q_3\).  Put
\[
 A=q_3(\Pu),\qquad B=q_3(\Pv),\qquad
 Z=b_3(\Pu,\Pv)\in\mathbb R.
\]
Theorem~\ref{thm:strong}, applied to every positive spectral pair,
gives
\[
 \alpha^2A+\beta^2B+2\alpha\beta Z\geq0
 \quad(\alpha,\beta\geq0),
\]
and hence \(Z\geq-\sqrt{AB}\).  Costa Rico's opposite-sign theorem
gives
\[
 \alpha^2A+\beta^2B-2\alpha\beta Z\geq0
 \quad(\alpha,\beta>0),
\]
and hence \(Z\leq\sqrt{AB}\).  Thus \(|Z|\leq\sqrt{AB}\), and
\[
 q_3(C)
 =|z_1|^2A+|z_2|^2B
  +2\Real(\overline{z_1}z_2)Z
 \geq
 \bigl(|z_1|\sqrt A-|z_2|\sqrt B\bigr)^2\geq0.
\]
\end{proof}

\begin{corollary}[Common initial--final two-plane]
\label{thm:common-plane}
If
\begin{equation}
 \dim\bigl(\operatorname{ran}C+\operatorname{ran}C^\dagger\bigr)\leq2,
 \label{eq:common-plane}
\end{equation}
then \(q_3(C)\geq0\).
\end{corollary}

\begin{proof}
Let \(E\) be a subspace of dimension at most two containing both ranges, and let \(P_E\) be its orthogonal projection.  Then
\[
 C=P_ECP_E.
\]
Writing
\[
 H=\frac{C+C^\dagger}{2},\qquad
 K=\frac{C-C^\dagger}{2i},
\]
we see that the Hermitian operators \(H,K\) are supported on \(E\) and therefore have rank at most two.  Every partial trace preserves adjoints, and for Hermitian \(X,Y\),
\[
 \|X+iY\|_2^2=\|X\|_2^2+\|Y\|_2^2
\]
because \(\Tr(XY-YX)=0\).  Applying this identity to every term of \eqref{eq:q3} gives \(q_3(C)=q_3(H)+q_3(K)\).  Both summands are nonnegative by Corollary~\ref{cor:normal-rank2}.
\end{proof}

\section{The residual nonnormal obstruction}
\label{sec:nonnormal}

\subsection{The crossed-Gram criterion}

Normality aligns the input and output directions within one spectral two-plane.  Once this alignment is lost, the new obstruction is the crossed pairing of two output rays with two input rays.  The singular-value decomposition isolates all of this genuinely nonnormal interference in one \(2\times2\) determinant.

\begin{proposition}[Exact crossed-Gram reduction]
\label{prop:crossed-gram}
Let \(u_1,u_2\) and \(v_1,v_2\) be orthonormal pairs and put
\(A_r=\ketbra{u_r}{v_r}\).  Then
\[
 q_3(z_1A_1+z_2A_2)\geq0\quad\text{for all }z_1,z_2\in\mathbb C
\]
if and only if
\begin{equation}
 G(u_1,u_2;v_1,v_2)=
 \begin{pmatrix}
  q_3(A_1)&b_3(A_1,A_2)\\
  \overline{b_3(A_1,A_2)}&q_3(A_2)
 \end{pmatrix}\succeq0.
 \label{eq:crossed-gram}
\end{equation}
Equivalently, the remaining scalar inequality is
\begin{equation}
 |b_3(A_1,A_2)|^2\leq q_3(A_1)q_3(A_2).
 \label{eq:crossed-determinant}
\end{equation}
Thus \(q_3(C)\geq0\) for every \(\rank C\leq2\) if and only if
\eqref{eq:crossed-determinant} holds for every two orthonormal pairs.
\end{proposition}

\begin{proof}
Polarization gives
\[
 q_3(z_1A_1+z_2A_2)
 =
 \begin{pmatrix}\overline z_1&\overline z_2\end{pmatrix}
 G(u_1,u_2;v_1,v_2)
 \begin{pmatrix}z_1\\z_2\end{pmatrix}.
\]
This quadratic form is nonnegative for all coefficients exactly when \(G\succeq0\), which is equivalent to \eqref{eq:crossed-determinant} because its diagonal entries are positive.  Indeed, on two copies of the tripartite space define
\begin{equation}
 R=\bigotimes_{i=1}^3\left(I-\frac12F_i\right).
 \label{eq:R-positive}
\end{equation}
Every factor in \eqref{eq:R-positive} has eigenvalues \(1/2\) and
\(3/2\), so \(R\succeq I/8\).  The swap trick yields
\begin{align}
 q_3(A_r)
 &=\langle u_r\otimes v_r,R(u_r\otimes v_r)\rangle,\label{eq:R-diagonal}\\
 b_3(A_1,A_2)
 &=\langle u_1\otimes v_2,R(u_2\otimes v_1)\rangle.
 \label{eq:R-cross}
\end{align}
Hence \(q_3(A_r)\geq1/8\).  Finally, every rank-two operator has a
singular-value decomposition with precisely the two orthonormal
families appearing above, and arbitrary phases can be absorbed into
one of those families.
\end{proof}

For a fixed coefficient vector \((z_1,z_2)\), negativity of \(q_3(z_1A_1+z_2A_2)\) implies \(G\nsucceq0\), but \(G\nsucceq0\) need not imply negativity for that prescribed vector.  It does, however, guarantee negativity for some coefficient vector within the
same four rays; equivalently, this occurs when \eqref{eq:crossed-determinant} fails.

The crossed ordering carries the entire obstruction.  Let \(U,V:\mathbb C^2\to\bigotimes_i\mathcal H_i\) be the isometries \(U|r\rangle=u_r\), \(V|r\rangle=v_r\), and set
\[
 K=(U^\dagger\otimes V^\dagger)R(U\otimes V)\succeq0,
\]
Then \(G\) is the \(\{|11\rangle,|22\rangle\}\) principal submatrix of
\(K^{\Gamma_2}\), where \(\Gamma_2\) is the partial transpose on the
second auxiliary copy of \(\mathbb C^2\).  With
\(K_{rs,tu}:=\langle r,s|K|t,u\rangle\), this follows from
\[
 b_3(A_1,A_2)=K_{12,21}.
\]
Ordinary Cauchy--Schwarz for \(K\) instead gives
\[
 |b_3(A_1,A_2)|^2
 \leq q_3(\ketbra{u_1}{v_2})q_3(\ketbra{u_2}{v_1}),
\]
with the wrong two diagonal factors.  For \(u_1=v_1=|000\rangle\) and 
\(u_2=v_2=|111\rangle\), the desired diagonal product is \(1/64\), whereas the crossed one is \(1\).  In the same example the true off-diagonal entry is \((-1/2)^3=-1/8\), so \eqref{eq:crossed-determinant} is saturated.  This elementary example rules out a product-switch shortcut.

\subsection{Two broad nonnormal closures}

The first result constrains one local output--input support.  It is
independent of the common-plane corollary and covers every tripartite
system having at least one qubit-sized local factor.

\begin{theorem}[Local output--input support overlap]
\label{thm:local-two-support}
Let \(\xi=\operatorname{vec}C\) have Schmidt rank at most two across
the output--input cut.  Suppose that, for some site \(i\), there are
projections \(P_i,Q_i\) such that
\begin{equation}
 \xi=(P_i\otimes\overline{Q_i}\otimes I_{\widehat i})\xi,
 \label{eq:local-support}
\end{equation}
where the overline is entrywise conjugation in the fixed input basis,
\(\widehat i=\{1,2,3\}\setminus\{i\}\) labels the other two
output--input pairs, and \(I_{\widehat i}\) is their identity.
If
\begin{equation}
 \tau_i=\Tr(P_iQ_i)\leq2,
 \label{eq:local-overlap}
\end{equation}
then \(q_3(C)\geq0\).  In particular, this holds if either
\(\rank P_i\leq2\) or \(\rank Q_i\leq2\), and hence
\begin{equation}
 \min(d_1,d_2,d_3)\leq2,\quad\rank C\leq2
 \quad\Longrightarrow\quad q_3(C)\geq0.
 \label{eq:one-local-qubit}
\end{equation}
\end{theorem}

\begin{proof}
Under vectorization,
\begin{equation}
 q_3(C)=\langle\xi,W_3\xi\rangle,\qquad
 W_3=\bigotimes_{j=1}^3W_{d_j},\qquad
 W_d=I-\frac12\ketbra{\Omega_d}{\Omega_d}.
 \label{eq:W3}
\end{equation}
Put
\[
 {\mathsf P}=P_i\otimes\overline{Q_i},\qquad
 \mathcal K_i^{\mathrm{supp}}
 =\operatorname{ran}P_i\otimes
  \operatorname{ran}\overline{Q_i}.
\]
The conjugate on the input support is forced by vectorization:
\[
 (P_i\otimes\overline{Q_i})\operatorname{vec}X
 =\operatorname{vec}(P_iXQ_i).
\]
Since \(\xi=({\mathsf P}\otimes I_{\widehat i})\xi\), only the
compression of \(W_{d_i}\) to \(\mathcal K_i^{\mathrm{supp}}\)
contributes.  With
\[
 |\omega\rangle={\mathsf P}|\Omega_{d_i}\rangle,
\]
one has
\[
 {\mathsf P}W_{d_i}{\mathsf P}
 ={\mathsf P}-\frac12\ketbra{\omega}{\omega}
 \ \cong\
 I_{\mathcal K_i^{\mathrm{supp}}}
 -\frac12\ketbra{\omega}{\omega}
 =:G.
\]
The maximally entangled contraction identity gives
\[
 \|\omega\|^2
 =\langle\Omega_{d_i}|
 P_i\otimes\overline{Q_i}|\Omega_{d_i}\rangle
 =\Tr(P_i^T\overline{Q_i})
 =\Tr(P_iQ_i)=\tau_i,
\]
where \(\Tr(P_iQ_i)=\Tr(P_iQ_iP_i)\) is real and nonnegative.
Consequently,
\[
 \left\|G-I_{\mathcal K_i^{\mathrm{supp}}}\right\|_2
 =\frac12\bigl\|\ketbra{\omega}{\omega}\bigr\|_2
 =\frac{\tau_i}{2}\leq1.
\]
Moreover, \eqref{eq:local-overlap} gives \(G\succeq0\), because its only possibly nonunit eigenvalue is \(1-\tau_i/2\). The Gurvits--Barnum Hilbert--Schmidt ball theorem states that every operator \(I_{\mathcal K_i^{\mathrm{supp}}}+\Delta\succeq0\) with 
\(\Delta=\Delta^\dagger\) and \(\|\Delta\|_2\leq1\) is separable across the bipartition of \(\mathcal K_i^{\mathrm{supp}}\); the center is the unnormalized identity \cite[Theorem 1]{GurvitsBarnum}.  It follows
directly that \(G\) is separable.

Write a rank-one separable decomposition
\[
 G=\sum_\alpha
 \ketbra{a_\alpha}{a_\alpha}\otimes
 \ketbra{b_\alpha}{b_\alpha}.
\]
Contracting \(\xi\) with
\(\langle a_\alpha|\otimes\langle b_\alpha|\) leaves a vector \(\xi_\alpha\) on the other two output--input pairs whose Schmidt rank is still at most two.  Consequently,
\[
 \langle\xi,G\otimes W_{d_j}\otimes W_{d_k}\,\xi\rangle
 =\sum_\alpha
 \langle\xi_\alpha,W_{d_j}\otimes W_{d_k}\,\xi_\alpha\rangle
 \geq0
\]
by the two-copy endpoint theorem
\cite[Theorem A]{FraserEtAl}; see also the concurrent equal-dimension and complementary formulations
\cite{FuGaoPark,SongChen,BhartiGajjalaHaug}.
This proves the first assertion.  Since
\(\Tr(P_iQ_i)\leq\min(\rank P_i,\rank Q_i)\), either support-rank condition implies \eqref{eq:local-overlap}.  Taking \(P_i=Q_i=I\) when \(d_i\leq2\) proves
\eqref{eq:one-local-qubit}.
\end{proof}

For a singular-value decomposition, let
\[
 U=\operatorname{span}\{u_1,u_2\},\qquad
 V=\operatorname{span}\{v_1,v_2\},
\]
and let \(P_i,Q_i\) be their minimal one-site support projections,
equivalently the support projections of \(\Tr_{\widehat i}P_U\) and \(\Tr_{\widehat i}P_V\), where \(P_U\) and \(P_V\) are the orthogonal projections onto \(U\) and \(V\), respectively, and \(\Tr_{\widehat i}\) traces out the two sites other than \(i\).
Thus a counterexample, if one exists, must obey
\[
 \Tr(P_iQ_i)>2\qquad(i=1,2,3).
\]
In particular, both its output and input supports must have dimension at least three at every site.  In \(3\times3\times3\), all six local supports would have to be the full qutrit space.

The second closure rests on a switched Schur complement available when one support plane contains a product reference ray.  The matrix order in this complement is reversed from the ordinary Schur order.  We therefore keep the full co-Choi argument beside the theorem that uses it.

\begin{theorem}[A product ray in either support plane]
\label{thm:product-ray}
Let \(\rank C\leq2\).  If either \(\operatorname{ran}C\) or
\(\operatorname{ran}C^\dagger\) contains a nonzero fully product vector, then \(q_3(C)\geq0\).
\end{theorem}

The crucial input is the following product-reference lemma.

\begin{lemma}[Product-reference switched Schur complement]
\label{lem:product-ray-schur}
Let \(u=u_1\otimes u_2\otimes u_3\) be a normalized product vector and
let \(x\) be arbitrary.  Compress the positive operator \(R\) in
\eqref{eq:R-positive} on its first copy and define operators on the
second copy by
\[
 A=(\langle u|\otimes I)R(|u\rangle\otimes I),\quad
 L=(\langle u|\otimes I)R(|x\rangle\otimes I),\quad
 D=(\langle x|\otimes I)R(|x\rangle\otimes I).
\]
Then \(A\succ0\) and
\begin{equation}
 D-LA^{-1}L^\dagger\succeq0.
 \label{eq:product-ray-schur}
\end{equation}
The ordinary Schur complement of the compression of \(R\) contains \(L^\dagger A^{-1}L\), whereas the product structure proves the switched order required below.
\end{lemma}

\begin{proof}
Local unitary covariance permits \(u_i=|0\rangle\).  On one local space put \(P=|0\rangle\langle0|\), \(Q=I-P\), and
\[
 A_0=I-\frac12P,\qquad
 \ell_a=(\langle0|\otimes I)
 \left(I-\frac12F\right)(|a\rangle\otimes I)
 =\delta_{a0}I-\frac12|a\rangle\langle0|.
\]
For the chosen local basis let \(E_{ab}=|a\rangle\langle b|\).
Define maps on these matrix units by
\[
 \Lambda(E_{ab})=\delta_{ab}I-\frac12E_{ab},\qquad
 {\cal M}(E_{ab})=\ell_aA_0^{-1}\ell_b^\dagger .
\]
Writing \(x=\sum_{\boldsymbol a}x_{\boldsymbol a}|\boldsymbol a\rangle\)
and expanding the tensor product gives
\begin{align}
 A&=A_{0,1}\otimes A_{0,2}\otimes A_{0,3},\notag\\
 D&=(\Lambda_1\otimes\Lambda_2\otimes\Lambda_3)
       (|x\rangle\langle x|),\notag\\
 LA^{-1}L^\dagger
 &=({\cal M}_1\otimes{\cal M}_2\otimes{\cal M}_3)
       (|x\rangle\langle x|).
 \label{eq:product-ray-map-identities}
\end{align}

For a linear map \(\mathcal E\), let
\[
 J(\mathcal E)=\sum_{a,b}E_{ab}\otimes\mathcal E(E_{ab}),\qquad
 \widehat J(\mathcal E)=\sum_{a,b}E_{ba}\otimes\mathcal E(E_{ab}).
\]
If \({\mathsf T}\) denotes matrix transposition in the chosen basis,
then
\begin{equation}
 \widehat J(\mathcal E)=J(\mathcal E\circ{\mathsf T}).
 \label{eq:cochoi-ccp}
\end{equation}
Thus \(\widehat J(\mathcal E)\succeq0\) is equivalent to complete copositivity of \(\mathcal E\), meaning that
\(\mathcal E\circ{\mathsf T}\) is completely positive.  Set
\[
 C_0=\widehat J(\Lambda)=I-\frac12F,\qquad
 S=\widehat J({\cal M}).
\]
The mutually orthogonal spaces \(P\otimes P\),
\(\operatorname{span}\{|0a\rangle,|a0\rangle\}\) for \(a\ne0\),
and \(Q\otimes Q\) reduce both matrices.  On the first space
\(C_0=S=1/2\); on the middle spaces, in the displayed ordered basis,
\[
 C_0=\begin{pmatrix}1&-1/2\\-1/2&1\end{pmatrix},
 \qquad
 S=\begin{pmatrix}1&-1/2\\-1/2&0\end{pmatrix};
\]
and on \(Q\otimes Q\),
\[
 C_0=I-\frac12F_Q,\qquad S=\frac12F_Q.
\]
Here \(F_Q\) is the swap restricted to
\(\operatorname{ran}Q\otimes\operatorname{ran}Q\).
Thus \(C_0-S\succeq0\) and \(C_0+S\succeq0\) on every reducing
space, or
\begin{equation}
 -C_0\preceq S\preceq C_0.
 \label{eq:cochoi-order}
\end{equation}
Since \(C_0\succ0\), the Hermitian matrix
\({\mathsf R}=C_0^{-1/2}SC_0^{-1/2}\) is a contraction.  After the
canonical unitary permutation that groups like tensor factors,
\[
 \widehat J\!\left(
 \bigotimes_i\Lambda_i-\bigotimes_i{\cal M}_i\right)
 =
 \left(\bigotimes_i C_{0,i}^{1/2}\right)
 \left(I-\bigotimes_i{\mathsf R}_i\right)
 \left(\bigotimes_i C_{0,i}^{1/2}\right)\succeq0.
\]
The subscript \(i\) here denotes the copy of each one-site object
\(\Lambda,\mathcal M,C_0,\mathsf R,A_0\) acting at site \(i\).
Indeed, every eigenvalue of \(\bigotimes_i{\mathsf R}_i\) lies in \([-1,1]\).  By \eqref{eq:cochoi-ccp}, the map in parentheses is completely copositive and therefore positive.  Applying it to \(|x\rangle\langle x|\) in \eqref{eq:product-ray-map-identities}
proves \eqref{eq:product-ray-schur}.
\end{proof}

\begin{proof}[Proof of Theorem~\ref{thm:product-ray}]
Suppose first that \(\rank C=2\) and the range contains a product unit vector \(u\).  Complete it to an orthonormal basis \(u,x\) of the range and write
\[
 C=|u\rangle\langle v|+|x\rangle\langle y|,
\]
where \(v,y\) need not be orthogonal.  For this \(u,x\), define
\[
 A=(\langle u|\otimes I)R(|u\rangle\otimes I),\quad
 L=(\langle u|\otimes I)R(|x\rangle\otimes I),\quad
 D=(\langle x|\otimes I)R(|x\rangle\otimes I).
\]
Lemma~\ref{lem:product-ray-schur} gives
\(D-LA^{-1}L^\dagger\succeq0\).  Set
\[
 a=\langle v,Av\rangle,\qquad
 d=\langle y,Dy\rangle,\qquad
 g=\langle y,Lv\rangle.
\]
Here \(a\) and \(d\) are the two diagonal values of \(q_3\), and \(g\)
is their polarized cross term.  Cauchy--Schwarz and
\eqref{eq:product-ray-schur} give
\[
 |g|^2\leq
 a\,\langle y,LA^{-1}L^\dagger y\rangle\leq ad.
\]
Consequently,
\[
 q_3(C)=a+d+2\Real g
 \geq(\sqrt a-\sqrt d)^2\geq0.
\]
The assertion for a product ray in \(\operatorname{ran}C^\dagger\)
follows from \(q_3(C^\dagger)=q_3(C)\); the rank-one case follows by setting one covector to zero.
\end{proof}

The product-ray and local-overlap closures are independent.  For example, in \((\mathbb C^3)^{\otimes3}\), let
\[
 u=|000\rangle,\quad
 x=\frac{|012\rangle+|120\rangle+|201\rangle}{\sqrt3},\quad
 v=\frac{|000\rangle+|111\rangle+|222\rangle}{\sqrt3},\quad y=x,
\]
and \(C=|u\rangle\langle v|+|x\rangle\langle y|\).  All one-site supports are the full qutrit space, while \(\dim(\operatorname{ran}C+\operatorname{ran}C^\dagger)=3\) and
\(CC^\dagger\ne C^\dagger C\).  Thus neither the local-overlap theorem nor the common-plane corollary applies, whereas Theorem~\ref{thm:product-ray} does.

Two more specialized closures---a local-traceless face and operators between distinct locally orthogonal product codes---are proved in Appendix~\ref{app:additional-sectors}.

\section{Sharpness and the rank boundary}
\label{sec:sharpness}

The following constructions separate two issues: optimality of the constants within rank two and necessity of the rank-two hypothesis.

\subsection{A nonnormal rank-two/rank-three boundary family}

We continue to write \(E_{ab}=|a\rangle\langle b|\) for matrix units.

\begin{proposition}[Nonnormal equality at rank two and negativity at rank three]
\label{prop:nonnormal-boundary}
Let \(m\geq2\), assume the local dimensions contain a
\(2\times2\times m\) subspace, and define
\begin{equation}
 C_m=E_{01}\otimes E_{01}\otimes\frac{I_m}{\sqrt m}.
 \label{eq:Cm}
\end{equation}
Then \(C_m\) is nonnormal, \(\rank C_m=m\), \(\|C_m\|_2=1\), and
\begin{equation}
 q_3(C_m)=1-\frac m2.
 \label{eq:Cm-q}
\end{equation}
In particular, \(C_2\) is a genuinely nonnormal rank-two equality
operator, whereas \(q_3(C_3)=-1/2\).
\end{proposition}

\begin{proof}
For product operators the form factorizes:
\[
 q_3(X_1\otimes X_2\otimes X_3)
 =\prod_{i=1}^3
 \left(\|X_i\|_2^2-\frac12|\Tr X_i|^2\right).
\]
The first two factors of \eqref{eq:Cm} have Hilbert--Schmidt norm one and trace zero, while \(I_m/\sqrt m\) has Hilbert--Schmidt norm one and trace \(\sqrt m\).  This proves \eqref{eq:Cm-q}.  The rank and norm statements are immediate.  Finally,
\[
 C_mC_m^\dagger
 =E_{00}\otimes E_{00}\otimes\frac{I_m}{m},\qquad
 C_m^\dagger C_m
 =E_{11}\otimes E_{11}\otimes\frac{I_m}{m},
\]
so \(C_m\) is nonnormal.
\end{proof}

\subsection{Three-qubit anti-state equality}

On one qubit, let \(K\) denote conjugation in the computational basis, put
\[
 \sigma_y=\begin{pmatrix}0&-i\\ i&0\end{pmatrix},
 \qquad J=i\sigma_y,
\]
and define the antiunitary
\(\vartheta=JK\).  If \(K_{\rm global}\) denotes component-wise conjugation of the full three-qubit coefficient vector, set
\[
 \Theta=J^{\otimes3}K_{\rm global}
 =\vartheta^{\otimes3}.
\]
Here the tensor product of local antiunitaries is understood canonically.  The local map \(\vartheta=i\sigma_yK\) is the spin flip entering concurrence and universal state inversion
\cite{Wootters,RungtaEtAl}; its odd-party tensor power is the anti-state map used in \cite{TranEtAl}.

\begin{proposition}[Anti-state saturation]
\label{prop:antistate}
For every normalized three-qubit \(u\), let \(v=\Theta u\).  Then
\[
 u\perp v,\qquad
 3p_3(u,v)-p_2(u,v)=\sqrt{a(u)a(v)}.
\]
Moreover, for all \(\lambda,\mu\geq0\),
\begin{equation}
 q_3(\lambda\Pu+\mu\Pv)
 =(\lambda-\mu)^2\left(\frac18+a(u)\right).
 \label{eq:antistate-q}
\end{equation}
\end{proposition}

\begin{proof}
Since \(\vartheta^2=-I\), one has \(\Theta^2=-I\), and Kramers orthogonality gives \(\langle u,\Theta u\rangle=0\).  Each local antisymmetric two-qubit space is spanned by the singlet.  Contracting \(u\otimes\Theta u\) with the three normalized singlets gives
\[
 \langle\omega|^{\otimes3}(u\otimes\Theta u)
 =-\frac{1}{2\sqrt2},\qquad
 |\omega\rangle=\frac{|01\rangle-|10\rangle}{\sqrt2},
 \]
where the phase depends on the convention for \(\vartheta\), but the modulus does not.  Indeed, the one-qubit contraction is
\(\langle\omega|(x\otimes\vartheta y)
=-\langle y,x\rangle/\sqrt2\).  Consequently,
\[
 p_3(u,\Theta u)=\frac18.
\]
The spin flip maps each one-qubit reduction to
\(\rho_i^v=I-\rho_i^u\).  Explicitly, if
\[
 \rho=\begin{pmatrix}a&c\\\overline c&1-a\end{pmatrix},
 \qquad
 J=\begin{pmatrix}0&1\\-1&0\end{pmatrix},
\]
then
\[
 J\overline\rho\,J^\dagger
 =\begin{pmatrix}1-a&-c\\-\overline c&a\end{pmatrix}
 =I-\rho .
\]
Consequently,
\[
 \sum_i\Tr(\rho_i^u\rho_i^v)
 =\sum_i\left(1-\Tr[(\rho_i^u)^2]\right)=4a(u).
\]
Equation~\eqref{eq:p2p3}, together with \(p_3=1/8\), gives
\[
 p_2=\frac38-a(u),\qquad a(v)=a(u).
\]
Therefore \(3p_3-p_2=a(u)\), proving equality in \eqref{eq:sharp-plucker}.  Substituting the same identities into \eqref{eq:S-exterior} and then \eqref{eq:q-S} proves
\eqref{eq:antistate-q}.
\end{proof}

Thus every balanced anti-state mixture is an equality state, and the constant in Theorem~\ref{thm:sharp-plucker} cannot be reduced.
For the product antipodes \(u=|000\rangle\), \(v=|111\rangle\), one has \(a(u)=a(v)=0\), and \eqref{eq:antistate-q} reduces to
\[
 q_3(\lambda P_u+\mu P_v)=\frac{(\lambda-\mu)^2}{8}.
\]
This attains equality in the positive lower bound \eqref{eq:strong-bound}, so its coefficient is sharp as well.

\subsection{Failure at rank three}

\begin{proposition}[Rank-three counterexample]
\label{prop:rank3-counterexample}
There is a positive tripartite state \(C\) of rank three for which both
\(\cS(C)<0\) and \(q_3(C)<0\).
\end{proposition}

\begin{proof}
Let
\[
 C=\ketbra{\psi}{\psi}_{12}\otimes\tau_3,\qquad
 e=\Tr(\psi_1^2)=\Tr(\psi_2^2),\qquad t=\Tr\tau_3^2.
\]
Here
\(\psi_1=\Tr_2\ketbra{\psi}{\psi}\) and
\(\psi_2=\Tr_1\ketbra{\psi}{\psi}\), while \(\tau_3\) is a density
operator on subsystem \(3\).
The global, one-body, and two-body purity sums are
\[
 \Tr C^2=t,\qquad
 \sum_i\Tr C_i^2=2e+t,\qquad
 \sum_i\Tr C_{\bar i}^2=1+2et.
\]
Consequently
\begin{equation}
 \cS(C)=2(2t-1)(1-e).
 \label{eq:counter-family}
\end{equation}
Choose
\[
 |\Phi_3\rangle=\frac1{\sqrt3}\sum_{a=0}^2|aa\rangle,
\]
maximally entangled on \(\mathbb C^3\otimes\mathbb C^3\), and set
\(\psi=\Phi_3\), \(\tau_3=I_3/3\).
Then \(e=t=1/3\), \(\rank C=3\), and
\[
 \cS(C)=-\frac49,\qquad
 q_3(C)=\frac14\cS(C)+\frac18(2t-1)
       =-\frac{11}{72}.
\]
\end{proof}

This example lies outside the Schmidt-rank-at-most-two witness domain in \eqref{eq:qk-criterion}.  It blocks an arbitrary-rank extension of Theorem~\ref{thm:strong} and of the associated linear-entropy inequality.  Because it is not an admissible rank-two witness, it has no direct implication for Werner-state distillability.

Within the tensor-product family in the proof, rank two is the exact threshold for \(\cS\): if \(\rank\tau\leq2\), then \(t\geq1/2\), making \eqref{eq:counter-family} nonnegative.

\section{Operational and physical meaning}
\label{sec:physical-meaning}

Equation~\eqref{eq:qk-criterion} identifies \(C\) with the coefficient matrix of a Schmidt-rank-two test vector.  The main results impose three successive filters on a negative witness at \(\alpha=-1/2\): it cannot be normal (including positive semidefinite); its global
initial and final two-planes must be distinct and product-ray-free; and their local output--input overlap must exceed two at every site. Appendix~\ref{app:additional-sectors} adds two independent exclusions, namely every local-traceless face and the crossed locally orthogonal
product-code class.  Each restriction concerns the coherent test vector; none assumes that the Werner state has low rank.

The symbol \(C\) has two uses.  In the distillation criterion it is an arbitrary coefficient matrix and need not be Hermitian.  If \(C\succeq0\) and \(\Tr C=1\), the same matrix is also a tripartite density operator.  In that interpretation, \(\cS(C)\geq0\) is exactly the linear-entropy relation \eqref{eq:sssa}.  Proposition \ref{prop:rank3-counterexample} shows that this entropy statement is rank-sensitive and fails as a universal entropy law.

\subsection{Local SWAP parity and measurable purities}

For density operators \(\rho,\sigma\) on a multipartite system, let \(F_i\) exchange the two copies of subsystem \(i\), and define \(F_S=\prod_{i\in S}F_i\), with \(F_\varnothing=I\).  Also write \(\rho_S=\Tr_{\bar S}\rho\) and \(\sigma_S=\Tr_{\bar S}\sigma\).  The SWAP identity reads
\begin{equation}
 \Tr[(\rho\otimes\sigma)F_S]
 =\Tr(\rho_S\sigma_S).
 \label{eq:physical-swap}
\end{equation}
When \(\rho=\sigma\), the right side is the subsystem purity \(\Tr\rho_S^2\).  Controlled-SWAP circuits, two-copy beam-splitter
interference, and randomized-measurement protocols provide
experimentally established routes to these nonlinear observables
\cite{EkertEtAl,AlvesEtAl,BovinoEtAl,IslamEtAl,ElbenEtAl,BrydgesEtAl,
MillerEtAl}.
Because the local swaps commute, the complete joint parity distribution is obtained from the correlators in \eqref{eq:physical-swap} by the Walsh transform written explicitly in
\eqref{eq:walsh-sector}.  The quantities \(a,p_2,p_3\) are particular coarse-grainings of this distribution. Antisymmetric sectors themselves also carry nontrivial
entanglement-cost and distillable-key bounds \cite{ChristandlSchuchWinter}; here they enter specifically as outcomes of copy-exchange parity measurements.

For a pure tripartite vector \(w\),
\[
 4a(w)=\sum_{i=1}^3
 \bigl(1-\Tr[(\rho_i^w)^2]\bigr).
\]
Each summand is the linear entropy of entanglement across the bipartition \(i:\bar i\).  Theorem~\ref{thm:sharp-plucker} therefore states that the weighted parity combination \(3p_3-p_2\), comparing the fully antisymmetric outcome with the two-minus sectors, is bounded by
the geometric mean of the total one-versus-rest entanglement of the two rays.  This gives the Pl\"ucker inequality an operational two-copy reading: it is a sharp compatibility constraint on observable local exchange parities.

\subsection{Transition coherence and the second compound}

The operator
\[
 X_k=\Tr_k\ketbra{u}{v}
\]
is the coherence between two orthogonal rays that remains after subsystem \(k\) is discarded.  Theorem~\ref{thm:nuclear-pt} controls the total Hilbert--Schmidt strength of its two further reductions.  The quantity retained in the proof,
\[
 \|X_k\|_1^2-\|X_k\|_2^2
 =2\|\wedgeop X_k\|_1,
\]
is the sum of all pairwise products of singular values.  It vanishes
exactly when \(X_k\) has rank at most one and hence detects the irreducibly two-dimensional part of the reduced coherence.  The exterior-square estimate bounds that component by the geometric mean of the two local linear entropies.  This is the compensation mechanism
that is absent from a treatment of the SWAP-sector weights as unconstrained classical probabilities.

\section{Discussion}

A three-step dimension-free argument drives the results.  The double-antisymmetric projector first gives the exact Schmidt-rank-two overlap \(1/2\).  The SVD then converts that projection
statement into the partial-trace inequality of Theorem~\ref{thm:nuclear-pt}.  Finally, the second compound matrices of the three transition operators convert the nuclear-norm excess into local linear entropies.  The three cutwise estimates sum exactly to the Pl\"ucker combination in \eqref{eq:sharp-plucker}.

This chain closes every positive semidefinite rank-two test and, after combining the two spectral signs, every normal rank-two test in arbitrary finite local dimensions.  The anti-state family proves that the Pl\"ucker constant and the positive-spectrum square are sharp.
Proposition~\ref{prop:nonnormal-boundary} supplies a genuinely nonnormal rank-two equality and a rank-three negative continuation, while Proposition~\ref{prop:rank3-counterexample} rules out an arbitrary-rank extension of the positive result.

Consequently, any negative rank-two endpoint witness, if exists, must be genuinely nonnormal.  Its initial and final two-planes must be distinct and product-ray-free; their local output--input overlap must exceed two at every site; no local partial trace may vanish; and the operator must lie outside the crossed locally orthogonal product-code sector.  This necessary profile is the geometric content of the sector theorems.

Beyond normality, Proposition~\ref{prop:crossed-gram} reduces the unrestricted question to one crossed-Gram determinant.  The local support and product-ray theorems remove two broad geometric regions.
Even a future closure of the complete three-copy problem would remain a finite-copy result: for every prescribed copy number there exist states whose distillation requires more copies \cite{WatrousCopies}.  The asymptotic NPT bound-entanglement problem is therefore strictly
stronger.  The present results nevertheless convert the first unresolved three-copy layer from an unstructured rank-two search into a two-plane compatibility problem with sharp analytic boundaries.

\appendix

\section{Vectorization and the Werner distillability form}
\label{app:werner-criterion}

We record the contraction behind \eqref{eq:qk-criterion}.  Fix product
bases on
\[
 A=A_1\otimes\cdots\otimes A_k,\qquad
 B=B_1\otimes\cdots\otimes B_k,
\]
with \(A_j\simeq B_j\simeq\mathbb C^d\).  For
\(J\subseteq[k]=\{1,\ldots,k\}\), let
\[
 R_J=
 \bigotimes_{j\in J}\ketbra{\Omega_d}{\Omega_d}_{A_jB_j}
 \bigotimes_{j\notin J}I_{A_jB_j}.
\]
A direct index contraction gives
\begin{equation}
 \langle\operatorname{vec}C|R_J|\operatorname{vec}C\rangle
 =\|\Tr_JC\|_2^2.
 \label{eq:vec-contraction}
\end{equation}
Indeed, each maximally entangled projector identifies the row and column index at the corresponding tensor factor and sums over that common index; the remaining free indices are exactly the matrix indices of \(\Tr_JC\).  The cases \(J=\varnothing\) and \(J=[k]\)
reduce respectively to \(\|C\|_2^2\) and \(|\Tr C|^2\).

Since
\[
 (\rho_\alpha^{(d)})^\Gamma
 =\frac{I+\alpha\ketbra{\Omega_d}{\Omega_d}}
 {d^2+\alpha d},
\]
expanding its \(k\)-fold tensor power and applying
\eqref{eq:vec-contraction} proves \eqref{eq:qk-criterion}.  Finally, the standard distillability theorem states that \(\rho\) is \(k\)-copy distillable exactly when \((\rho^\Gamma)^{\otimes k}\) has negative expectation on a vector of Schmidt rank at most two
\cite{Horodecki1997,Horodecki1998,DurEtAl}.  Vectorization converts that Schmidt-rank condition into \(\rank C\leq2\).

\section{Swap and transition identities}
\label{app:swap}

For \(S\subseteq[k]\), let \(F_S=\prod_{i\in S}F_i\) be the product of the local swaps between the two copies, with \(F_\varnothing=I\), and let \(\bar S=S^c=[k]\setminus S\).  For operators on a tensor product, the swap trick gives
\[
 \Tr[(A\otimes B)F_S]
 =\Tr[(\Tr_{\bar S}A)(\Tr_{\bar S}B)]
\]
in its bilinear form, whereas the Hilbert--Schmidt form is
\[
 \Tr[(A^\dagger\otimes B)F_S]
 =\Tr\!\left[
   (\Tr_{\bar S}A)^\dagger\Tr_{\bar S}B
 \right].
\]
Consequently, for an arbitrary operator \(C\),
\[
 \|\Tr_JC\|_2^2
 =\Tr[(C^\dagger\otimes C)F_{J^c}].
\]
Hence, for every tripartite \(C\),
\[
 q_3(C)=
 \Tr\!\left[(C^\dagger\otimes C)
 \prod_{i=1}^3\left(F_i-\frac12I\right)\right].
\]
For Hermitian \(C\), \(C^\dagger\otimes C=C\otimes C\).

For normalized \(u,v\),
\[
 \|\Tr_S\ketbra{u}{v}\|_2^2
 =\Tr(\rho_S^u\rho_S^v).
\]
Expanding
\(\prod_i(I\pm F_i)/2\), using
\(\langle u,v\rangle=0\), and grouping the sectors with exactly two or three minus signs yields \eqref{eq:p2p3}.  Repeating the expansion for \(u\otimes u\), for which only even local parity survives, yields \eqref{eq:a-purity}.

For completeness, the kernel identity \eqref{eq:swap-kernel} follows term by term from
\begin{align*}
 \langle\Tr_{\mathcal K}\ketbra{u}{v},
         \Tr_{\mathcal K}\ketbra{s}{t}\rangle_{\HS}
 &=\langle u\otimes t,F_{\mathcal L}(v\otimes s)\rangle,\\
 \langle\Tr_{\mathcal L}\ketbra{u}{v},
         \Tr_{\mathcal L}\ketbra{s}{t}\rangle_{\HS}
 &=\langle u\otimes t,F_{\mathcal K}(v\otimes s)\rangle,\\
\overline{\Tr\ketbra{u}{v}}\Tr\ketbra{s}{t}
 &=\langle u\otimes t,v\otimes s\rangle.
\end{align*}

\section{Additional closed nonnormal sectors}
\label{app:additional-sectors}

The following two results are independent of the main nonnormal
closures.  
\subsection{The local-traceless sector}

For \(S\subseteq\{1,2,3\}\), let
\[
 q_S(X)=
 \left\langle X,
 \prod_{i\in S}\left(\id-\frac12\Tr_i^*\Tr_i\right)X
 \right\rangle_{\HS},
\]
where \(\Tr_i^*\) is the Hilbert--Schmidt adjoint of \(\Tr_i\).
Explicitly, up to the canonical permutation that restores the original
factor order, \(\Tr_i^*(Y)=I_i\otimes Y\).
Thus \(q_{\{1,2,3\}}=q_3\), while a two-element \(S\) gives the
two-copy endpoint form on those sites, acting blockwise on the
spectator factor.

\begin{theorem}[Local-traceless rank-two sector]
\label{cor:local-traceless}
If \(\rank C\leq2\) and \(\Tr_iC=0\) for at least one site \(i\),
then \(q_3(C)\geq0\).
\end{theorem}

\begin{proof}
If \(\Tr_iC=0\), then
\((\id-\frac12\Tr_i^*\Tr_i)C=C\).  Hence
\[
 q_3(C)=q_{\{1,2,3\}\setminus\{i\}}(C).
\]
Choose an orthonormal basis \(\{|a\rangle\}\) at the spectator site
\(i\), and define the explicit block compressions
\[
 C_{ab}=(\langle a|\otimes I)C(|b\rangle\otimes I).
\]
The two-copy form acts blockwise, and the matrix units
\(|a\rangle\langle b|\) are Hilbert--Schmidt orthonormal, so
\[
 q_{\{1,2,3\}\setminus\{i\}}(C)
 =\sum_{a,b}q_{\{1,2,3\}\setminus\{i\}}(C_{ab}).
\]
Each \(C_{ab}\) is a compression of \(C\), hence
\(\rank C_{ab}\leq\rank C\leq2\).  Every summand is nonnegative by the arbitrary-dimension two-copy endpoint theorem
\cite[Theorem A]{FraserEtAl}.
\end{proof}

\subsection{Crossed locally orthogonal product codes}

\begin{theorem}[Crossed locally orthogonal code sector]
\label{thm:triorthogonal}
Let \(V:\mathbb C^m\to\bigotimes_{i=1}^3\cH_i\) and \(W:\mathbb C^n\to\bigotimes_{i=1}^3\cH_i\) be the isometries
\[
 V|r\rangle=\bigotimes_{i=1}^3|e_r^{(i)}\rangle,\qquad
 W|s\rangle=\bigotimes_{i=1}^3|f_s^{(i)}\rangle,
\]
where each family \(\{e_r^{(i)}\}_{r=1}^m\) and \(\{f_s^{(i)}\}_{s=1}^n\) is orthonormal at every site \(i\).
No relation between the \(e\)- and \(f\)-families is assumed.  Put
\[
 S_i(r,s)=\langle f_s^{(i)},e_r^{(i)}\rangle,\qquad
 T=S_1\circ S_2\circ S_3,
\]
where \(\circ\) is the Hadamard product, and
\[
 w_{rs}=1-\frac12\sum_i|S_i(r,s)|^2
 +\frac14\sum_{i<j}|S_i(r,s)S_j(r,s)|^2.
\]
Then, for every \(M\in\mathbb C^{m\times n}\),
\begin{equation}
 q_3(VMW^\dagger)
 =\sum_{r,s}w_{rs}|M_{rs}|^2
 -\frac18\left|\sum_{r,s}T_{rs}M_{rs}\right|^2.
 \label{eq:crossed-code-exact}
\end{equation}
Consequently,
\begin{equation}
 \rank M\leq2
 \quad\Longrightarrow\quad
 q_3(VMW^\dagger)\geq
 \frac{2-\sigma_1(T)^2-\sigma_2(T)^2}{8}\|M\|_2^2
 \geq0.
 \label{eq:crossed-code-bound}
\end{equation}
Here \(\sigma_1(T)\geq\sigma_2(T)\) are the two largest singular values, and \(\sigma_2(T)=0\) when \(T\) has fewer than two singular values.
\end{theorem}

\begin{proof}
Expanding \(C=VMW^\dagger\) gives
\[
 C=\sum_{r,s}M_{rs}
 \bigotimes_{i=1}^3
 |e_r^{(i)}\rangle\langle f_s^{(i)}|.
\]
For every proper subset \(J\subsetneq\{1,2,3\}\),
\[
 \Tr_J C
 =\sum_{r,s}M_{rs}
 \left(\prod_{i\in J}S_i(r,s)\right)
 \bigotimes_{k\notin J}
 |e_r^{(k)}\rangle\langle f_s^{(k)}|.
\]
Since \(J^c\neq\varnothing\), orthonormality at any site \(k\in J^c\) makes the remaining tensor factors Hilbert--Schmidt orthogonal for different ordered pairs \((r,s)\).  Therefore
\[
 \|\Tr_J C\|_2^2
 =\sum_{r,s}|M_{rs}|^2
 \prod_{i\in J}|S_i(r,s)|^2
 \qquad(J\subsetneq\{1,2,3\}).
\]
On the other hand,
\[
 \Tr C=\sum_{r,s}M_{rs}\prod_iS_i(r,s)
 =\sum_{r,s}T_{rs}M_{rs}.
\]
Substitution in \eqref{eq:q3} proves
\eqref{eq:crossed-code-exact}.

For \(0\leq a_i\leq1\), the multiaffine polynomial
\[
 1-\frac12\sum_i a_i+\frac14\sum_{i<j}a_i a_j
\]
has minimum \(1/4\) on \([0,1]^3\), as follows by checking its eight
vertices.  Hence \(w_{rs}\geq1/4\).

If \(\mathsf E_i,\mathsf G_i\) are the isometries with columns
\(e_r^{(i)}\) and \(f_s^{(i)}\), then
\(S_i=\overline{\mathsf E_i^\dagger\mathsf G_i}\), so every
\(S_i\) is a contraction.  With diagonal isometries
\(\Delta_k|r\rangle=|rrr\rangle\),
\[
 T=\Delta_m^\dagger
 (S_1\otimes S_2\otimes S_3)\Delta_n,
\]
and hence \(\|T\|_\infty\leq1\).  Finally,
\[
 \sum_{r,s}T_{rs}M_{rs}
 =\langle\overline T,M\rangle_{\HS}.
\]
If \(\rank M\leq2\), von Neumann's trace inequality gives
\[
 |\langle\overline T,M\rangle_{\HS}|^2
 \leq
 \bigl(\sigma_1(T)^2+\sigma_2(T)^2\bigr)\|M\|_2^2.
\]
Combining the last three estimates with
\eqref{eq:crossed-code-exact} proves
\eqref{eq:crossed-code-bound}.
\end{proof}

\begin{remark}
Taking \(V=W\) and the same local bases gives \(S_i=T=I_m\),
\(w_{rr}=1/4\), and \(w_{rs}=1\) for \(r\ne s\).  Formula
\eqref{eq:crossed-code-exact} then becomes
\[
 q_3(VMV^\dagger)
 =\|M\|_2^2-\frac34\sum_r|M_{rr}|^2
 -\frac18|\Tr M|^2,
\]
which is the same-code locally orthogonal formula.
\end{remark}

\begin{remark}[A fully supported square-zero example]
In \(3\times3\times3\), with labels understood modulo three, take
\[
 V|r\rangle=|r,r,r\rangle,\qquad
 W|s\rangle=|s,s+1,s+2\rangle
\]
and
\[
 M=
 \begin{pmatrix}
 1&0&1\\
 0&1&-1\\
 1&1&0
 \end{pmatrix}
 =
 \begin{pmatrix}1&0\\0&1\\1&1\end{pmatrix}
 \begin{pmatrix}1&0\\0&1\\1&-1\end{pmatrix}^{T}.
\]
The two code spaces are orthogonal, \(\rank M=2\), and the row and column two-planes of \(M\) have nonzero weight on all three labels. Thus \(C=VMW^\dagger\) is square zero, genuinely nonnormal, and has full output and input support at every local site.  In a locally  rthogonal code, a fully product vector is a single codeword, and all three local partial traces of \(C\) are nonzero.  Hence none of the preceding closed-sector theorems applies.  Here exactly one of the three local overlaps is nonzero for each ordered pair \((r,s)\), so
\(w_{rs}=1/2\), \(T=0\), and
\[
 q_3(C)=\frac12\|M\|_2^2=3.
\]
\end{remark}

\section*{Data availability}

Data sharing is not applicable to this article as no datasets were generated or analysed during the current study.


\begingroup
\small



\begin{thebibliography}{99}

\bibitem{BennettEtAl}
C.~H. Bennett, G.~Brassard, S.~Popescu, B.~Schumacher, J.~A. Smolin,
and W.~K. Wootters,
\emph{Purification of noisy entanglement and faithful teleportation via
noisy channels},
Phys.\ Rev.\ Lett.\ \textbf{76}, 722--725 (1996);
\href{https://arxiv.org/abs/quant-ph/9511027}{arXiv:quant-ph/9511027}.

\bibitem{Horodecki1997}
M.~Horodecki, P.~Horodecki, and R.~Horodecki,
\emph{Inseparable two spin-\(1/2\) density matrices can be distilled to
a singlet form},
Phys.\ Rev.\ Lett.\ \textbf{78}, 574--577 (1997);
\href{https://arxiv.org/abs/quant-ph/9607009}{arXiv:quant-ph/9607009}.

\bibitem{Horodecki1998}
M.~Horodecki, P.~Horodecki, and R.~Horodecki,
\emph{Mixed-state entanglement and distillation: Is there a ``bound''
entanglement in nature?},
Phys.\ Rev.\ Lett.\ \textbf{80}, 5239--5242 (1998);
\href{https://arxiv.org/abs/quant-ph/9801069}{arXiv:quant-ph/9801069}.

\bibitem{DurEtAl}
W.~D\"ur, J.~I. Cirac, M.~Lewenstein, and D.~Bru{\ss},
\emph{Distillability and partial transposition in bipartite systems},
Phys.\ Rev.\ A \textbf{61}, 062313 (2000);
\href{https://arxiv.org/abs/quant-ph/9910022}{arXiv:quant-ph/9910022}.

\bibitem{LewensteinPrimer}
M.~Lewenstein, D.~Bru{\ss}, J.~I. Cirac, B.~Kraus, M.~Ku\'s,
J.~Samsonowicz, A.~Sanpera, and R.~Tarrach,
\emph{Separability and distillability in composite quantum systems---a
primer},
J.\ Mod.\ Opt.\ \textbf{47}, 2481--2499 (2000);
\href{https://arxiv.org/abs/quant-ph/0006064}{arXiv:quant-ph/0006064}.

\bibitem{Peres}
A.~Peres,
\emph{Separability criterion for density matrices},
Phys.\ Rev.\ Lett.\ \textbf{77}, 1413--1415 (1996);
\href{https://arxiv.org/abs/quant-ph/9604005}{arXiv:quant-ph/9604005}.

\bibitem{HorodeckiPPT}
M.~Horodecki, P.~Horodecki, and R.~Horodecki,
\emph{Separability of mixed states: Necessary and sufficient conditions},
Phys.\ Lett.\ A \textbf{223}, 1--8 (1996);
\href{https://arxiv.org/abs/quant-ph/9605038}{arXiv:quant-ph/9605038}.

\bibitem{DiVincenzoEtAl}
D.~P. DiVincenzo, P.~W. Shor, J.~A. Smolin, B.~M. Terhal, and
A.~V. Thapliyal,
\emph{Evidence for bound entangled states with negative partial
transpose},
Phys.\ Rev.\ A \textbf{61}, 062312 (2000);
\href{https://arxiv.org/abs/quant-ph/9910026}{arXiv:quant-ph/9910026}.

\bibitem{WatrousCopies}
J.~Watrous,
\emph{Many copies may be required for entanglement distillation},
Phys.\ Rev.\ Lett.\ \textbf{93}, 010502 (2004);
\href{https://arxiv.org/abs/quant-ph/0312123}{arXiv:quant-ph/0312123}.

\bibitem{Clarisse}
L.~Clarisse,
\emph{The distillability problem revisited},
Quantum Inf.\ Comput.\ \textbf{6}, 539--560 (2006);
\href{https://arxiv.org/abs/quant-ph/0510035}{arXiv:quant-ph/0510035}.

\bibitem{ViannaDoherty}
R.~O. Vianna and A.~C. Doherty,
\emph{Study of the distillability of Werner states using entanglement
witnesses and robust semidefinite programs},
Phys.\ Rev.\ A \textbf{74}, 052306 (2006);
\href{https://arxiv.org/abs/quant-ph/0608095}{arXiv:quant-ph/0608095}.

\bibitem{PankowskiEtAl}
{\L}.~Pankowski, M.~Piani, M.~Horodecki, and P.~Horodecki,
\emph{A few steps more towards NPT bound entanglement},
IEEE Trans.\ Inf.\ Theory \textbf{56}, 4085--4100 (2010);
\href{https://arxiv.org/abs/0711.2613}{arXiv:0711.2613}.

\bibitem{Djokovic}
D.~Z. Djokovi\'c,
\emph{On two-distillable Werner states},
Entropy \textbf{18}, 216 (2016);
\href{https://arxiv.org/abs/1003.4337}{arXiv:1003.4337}.

\bibitem{HorodeckiProblems}
P.~Horodecki, {\L}.~Rudnicki, and K.~\.Zyczkowski,
\emph{Five open problems in quantum information theory},
PRX Quantum \textbf{3}, 010101 (2022);
\href{https://arxiv.org/abs/2002.03233}{arXiv:2002.03233}.

\bibitem{ShorSmolinTerhal}
P.~W. Shor, J.~A. Smolin, and B.~M. Terhal,
\emph{Nonadditivity of bipartite distillable entanglement follows from
a conjecture on bound entangled Werner states},
Phys.\ Rev.\ Lett.\ \textbf{86}, 2681--2684 (2001);
\href{https://arxiv.org/abs/quant-ph/0010054}{arXiv:quant-ph/0010054}.

\bibitem{MullerHermesReebWolf}
A.~M\"uller-Hermes, D.~Reeb, and M.~M. Wolf,
\emph{Positivity of linear maps under tensor powers},
J.\ Math.\ Phys.\ \textbf{57}, 015202 (2016);
\href{https://arxiv.org/abs/1502.05630}{arXiv:1502.05630}.

\bibitem{EggelingEtAl}
T.~Eggeling, K.~G.~H. Vollbrecht, R.~F. Werner, and M.~M. Wolf,
\emph{Distillability via protocols respecting the positivity of partial
transpose},
Phys.\ Rev.\ Lett.\ \textbf{87}, 257902 (2001);
\href{https://arxiv.org/abs/quant-ph/0104095}{arXiv:quant-ph/0104095}.

\bibitem{EckerEtAl}
S.~Ecker, P.~Sohr, L.~Bulla, M.~Huber, M.~Bohmann, and R.~Ursin,
\emph{Experimental single-copy entanglement distillation},
Phys.\ Rev.\ Lett.\ \textbf{127}, 040506 (2021);
\href{https://arxiv.org/abs/2101.11503}{arXiv:2101.11503}.

\bibitem{KanedaEtAl}
F.~Kaneda, R.~Shimizu, S.~Ishizaka, Y.~Mitsumori, H.~Kosaka, and
K.~Edamatsu,
\emph{Experimental activation of bound entanglement},
Phys.\ Rev.\ Lett.\ \textbf{109}, 040501 (2012);
\href{https://arxiv.org/abs/1111.6170}{arXiv:1111.6170}.

\bibitem{Werner}
R.~F. Werner,
\emph{Quantum states with Einstein--Podolsky--Rosen correlations
admitting a hidden-variable model},
Phys.\ Rev.\ A \textbf{40}, 4277--4281 (1989).

\bibitem{VollbrechtWerner}
K.~G.~H. Vollbrecht and R.~F. Werner,
\emph{Entanglement measures under symmetry},
Phys.\ Rev.\ A \textbf{64}, 062307 (2001);
\href{https://arxiv.org/abs/quant-ph/0010095}{arXiv:quant-ph/0010095}.

\bibitem{FangEtAl}
X.-X.~Fang, G.~N.~M. Tabia, K.-S.~Chen, Y.-C.~Liang, and H.~Lu,
\emph{Experimental single-copy distillation of quantumness from
higher-dimensional entanglement},
Phys.\ Rev.\ Lett.\ \textbf{134}, 150201 (2025);
\href{https://arxiv.org/abs/2410.06610}{arXiv:2410.06610}.

\bibitem{QianEtAl}
L.~Qian, L.~Chen, D.~Chu, and Y.~Shen,
\emph{A matrix inequality for entanglement distillation problem},
Linear Algebra Appl.\ \textbf{616}, 139--177 (2021);
\href{https://arxiv.org/abs/1908.02428}{arXiv:1908.02428}.

\bibitem{CostaRico}
P.~Costa Rico,
\emph{New partial trace inequalities and distillability of Werner
states},
Lett.\ Math.\ Phys.\ \textbf{115}, 47 (2025);
\href{https://arxiv.org/abs/2310.05726}{arXiv:2310.05726}.

\bibitem{QietAl}
S.-Y.~Qi, G.~Gupur, Y.-C.~Wu, and G.-P.~Guo,
\emph{Nonpositive-transpose entanglement and bound entanglement: From distillability sets to inequalities and multivariable insights.},
Phys.\ Rev.\ A \textbf{110}, 012406 (2024);
\href{https://arxiv.org/abs/2402.18037}{arXiv:2402.18037}.

\bibitem{BhartiGajjalaHaug}
K.~Bharti, R.~Gajjala, and T.~Haug,
\emph{Two-copy nondistillability of Werner states: sharp partial-trace
inequalities and finite-copy extensions},
\href{https://arxiv.org/abs/2607.24479}{arXiv:2607.24479} (2026).

\bibitem{FuGaoPark}
J.~Fu, L.~Gao, and S.-J.~Park,
\emph{A solution to 2-copy distillability of Werner states},
\href{https://arxiv.org/abs/2607.21367}{arXiv:2607.21367} (2026).

\bibitem{SongChen}
Z.~Song and L.~Chen,
\emph{A partial-trace matrix inequality and Werner-state
distillability},
\href{https://arxiv.org/abs/2607.23416}{arXiv:2607.23416} (2026).

\bibitem{FraserEtAl}
T.~C. Fraser, F.~Huber, B.~Pozsgay, and I.~Vona,
\emph{On the two-copy distillability of Werner states and a new
partial trace inequality},
\href{https://arxiv.org/abs/2607.24309}{arXiv:2607.24309} (2026).

\bibitem{Audenaert}
K.~M.~R. Audenaert,
\emph{Subadditivity of \(q\)-entropies for \(q>1\)},
J.\ Math.\ Phys.\ \textbf{48}, 083507 (2007);
\href{https://arxiv.org/abs/0705.1276}{arXiv:0705.1276}.

\bibitem{Rastegin}
A.~E. Rastegin,
\emph{Relations for certain symmetric norms and anti-norms before and
after partial trace},
J.\ Stat.\ Phys.\ \textbf{148}, 1040--1053 (2012);
\href{https://arxiv.org/abs/1202.3853}{arXiv:1202.3853}.

\bibitem{CostaRicoWolf}
P.~Costa Rico and M.~M. Wolf,
\emph{Partial trace relations beyond normal matrices},
\href{https://arxiv.org/abs/2507.18278}{arXiv:2507.18278} (2025).

\bibitem{Rudolph}
O.~Rudolph,
\emph{A separability criterion for density operators},
J.\ Phys.\ A: Math.\ Gen.\ \textbf{33}, 3951--3955 (2000);
\href{https://arxiv.org/abs/quant-ph/0002026}{arXiv:quant-ph/0002026}.

\bibitem{ChenWu}
K.~Chen and L.-A.~Wu,
\emph{A matrix realignment method for recognizing entanglement},
Quantum Inf.\ Comput.\ \textbf{3}, 193--202 (2003);
\href{https://arxiv.org/abs/quant-ph/0205017}{arXiv:quant-ph/0205017}.

\bibitem{MintertEtAl}
F.~Mintert, M.~Ku\'s, and A.~Buchleitner,
\emph{Concurrence of mixed multipartite quantum states},
Phys.\ Rev.\ Lett.\ \textbf{95}, 260502 (2005);
\href{https://doi.org/10.1103/PhysRevLett.95.260502}{doi:10.1103/PhysRevLett.95.260502};
\href{https://arxiv.org/abs/quant-ph/0411127}{arXiv:quant-ph/0411127}.

\bibitem{AolitaEtAl}
L.~Aolita, A.~Buchleitner, and F.~Mintert,
\emph{Scalable experimental estimation of multipartite entanglement},
Phys.\ Rev.\ A \textbf{78}, 022308 (2008);
\href{https://doi.org/10.1103/PhysRevA.78.022308}{doi:10.1103/PhysRevA.78.022308};
\href{https://arxiv.org/abs/0710.3529}{arXiv:0710.3529}.

\bibitem{EkertEtAl}
A.~K. Ekert, C.~M. Alves, D.~K.~L. Oi, M.~Horodecki, P.~Horodecki,
and L.~C. Kwek,
\emph{Direct estimations of linear and nonlinear functionals of a
quantum state},
Phys.\ Rev.\ Lett.\ \textbf{88}, 217901 (2002);
\href{https://arxiv.org/abs/quant-ph/0203016}{arXiv:quant-ph/0203016}.

\bibitem{AlvesEtAl}
C.~M. Alves, P.~Horodecki, D.~K.~L. Oi, L.~C. Kwek, and A.~K. Ekert,
\emph{Direct estimation of functionals of density operators by local
operations and classical communication},
Phys.\ Rev.\ A \textbf{68}, 032306 (2003);
\href{https://arxiv.org/abs/quant-ph/0304123}{arXiv:quant-ph/0304123}.

\bibitem{BovinoEtAl}
F.~A. Bovino, G.~Castagnoli, A.~Ekert, P.~Horodecki, C.~M. Alves, and
A.~V. Sergienko,
\emph{Direct measurement of nonlinear properties of bipartite quantum
states},
Phys.\ Rev.\ Lett.\ \textbf{95}, 240407 (2005);
\href{https://arxiv.org/abs/quant-ph/0511187}{arXiv:quant-ph/0511187}.

\bibitem{IslamEtAl}
R.~Islam, R.~Ma, P.~M. Preiss, M.~E. Tai, A.~Lukin, M.~Rispoli, and
M.~Greiner,
\emph{Measuring entanglement entropy in a quantum many-body system},
Nature \textbf{528}, 77--83 (2015);
\href{https://arxiv.org/abs/1509.01160}{arXiv:1509.01160}.

\bibitem{ElbenEtAl}
A.~Elben, B.~Vermersch, M.~Dalmonte, J.~I. Cirac, and P.~Zoller,
\emph{R\'enyi entropies from random quenches in atomic Hubbard and spin
models},
Phys.\ Rev.\ Lett.\ \textbf{120}, 050406 (2018);
\href{https://arxiv.org/abs/1709.05060}{arXiv:1709.05060}.

\bibitem{MillerEtAl}
D.~Miller, K.~Levi, L.~Postler, A.~Steiner, L.~Bittel,
G.~A.~L. White, Y.~Tang, E.~J. Kuehnke, A.~A. Mele, S.~Khatri,
L.~Leone, J.~Carrasco, C.~D. Marciniak, I.~Pogorelov,
M.~Guevara-Bertsch, R.~Freund, R.~Blatt, P.~Schindler, T.~Monz,
M.~Ringbauer, and J.~Eisert,
\emph{Experimental measurement and a physical interpretation of quantum
shadow enumerators},
Phys.\ Rev.\ Research \textbf{8}, 023318 (2026);
\href{https://arxiv.org/abs/2408.16914}{arXiv:2408.16914}.

\bibitem{Levay}
P.~L\'evay,
\emph{On the geometry of a class of \(N\)-qubit entanglement
monotones},
J.\ Phys.\ A: Math.\ Gen.\ \textbf{38}, 9075--9085 (2005);
\href{https://arxiv.org/abs/quant-ph/0507070}{arXiv:quant-ph/0507070}.

\bibitem{Harris}
J.~Harris,
\emph{Algebraic Geometry: A First Course},
Graduate Texts in Mathematics, Vol.~133 (Springer, New York, 1992).

\bibitem{Landsberg}
J.~M. Landsberg,
\emph{Tensors: Geometry and Applications},
Graduate Studies in Mathematics, Vol.~128
(American Mathematical Society, Providence, 2012).

\bibitem{ShorLaflamme}
P.~W. Shor and R.~Laflamme,
\emph{Quantum analog of the MacWilliams identities for classical coding
theory},
Phys.\ Rev.\ Lett.\ \textbf{78}, 1600--1602 (1997);
\href{https://arxiv.org/abs/quant-ph/9610040}{arXiv:quant-ph/9610040}.

\bibitem{RainsShadow}
E.~M. Rains,
\emph{Quantum shadow enumerators},
IEEE Trans.\ Inf.\ Theory \textbf{45}, 2361--2366 (1999);
\href{https://arxiv.org/abs/quant-ph/9611001}{arXiv:quant-ph/9611001}.

\bibitem{Rains}
E.~M. Rains,
\emph{Polynomial invariants of quantum codes},
IEEE Trans.\ Inf.\ Theory \textbf{46}, 54--59 (2000);
\href{https://arxiv.org/abs/quant-ph/9704042}{arXiv:quant-ph/9704042}.

\bibitem{EltschkaEtAl}
C.~Eltschka, F.~Huber, O.~G\"uhne, and J.~Siewert,
\emph{Exponentially many entanglement and correlation constraints for
multipartite quantum states},
Phys.\ Rev.\ A \textbf{98}, 052317 (2018);
\href{https://arxiv.org/abs/1807.09165}{arXiv:1807.09165}.

\bibitem{WyderkaGuhne}
N.~Wyderka and O.~G\"uhne,
\emph{Characterizing quantum states via sector lengths},
J.\ Phys.\ A: Math.\ Theor.\ \textbf{53}, 345302 (2020);
\href{https://arxiv.org/abs/1905.06928}{arXiv:1905.06928}.

\bibitem{JohnstonKribs}
N.~Johnston and D.~W. Kribs,
\emph{A family of norms with applications in quantum information
theory},
J.\ Math.\ Phys.\ \textbf{51}, 082202 (2010);
\href{https://arxiv.org/abs/0909.3907}{arXiv:0909.3907}.

\bibitem{JohnstonEtAl}
N.~Johnston, D.~W. Kribs, V.~I. Paulsen, and R.~Pereira,
\emph{Minimal and maximal operator spaces and operator systems in
entanglement theory},
J.\ Funct.\ Anal.\ \textbf{260}, 2407--2423 (2011);
\href{https://arxiv.org/abs/1010.1432}{arXiv:1010.1432}.

\bibitem{HornJohnson}
R.~A. Horn and C.~R. Johnson,
\emph{Matrix Analysis}, 2nd ed.
(Cambridge University Press, Cambridge, 2012).

\bibitem{GurvitsBarnum}
L.~Gurvits and H.~Barnum,
\emph{Largest separable balls around the maximally mixed bipartite
quantum state},
Phys.\ Rev.\ A \textbf{66}, 062311 (2002);
\href{https://arxiv.org/abs/quant-ph/0204159}{arXiv:quant-ph/0204159}.

\bibitem{Wootters}
W.~K. Wootters,
\emph{Entanglement of formation of an arbitrary state of two qubits},
Phys.\ Rev.\ Lett.\ \textbf{80}, 2245--2248 (1998);
\href{https://arxiv.org/abs/quant-ph/9709029}{arXiv:quant-ph/9709029}.

\bibitem{RungtaEtAl}
P.~Rungta, V.~Bu\v{z}ek, C.~M. Caves, M.~Hillery, and G.~J. Milburn,
\emph{Universal state inversion and concurrence in arbitrary
dimensions},
Phys.\ Rev.\ A \textbf{64}, 042315 (2001);
\href{https://arxiv.org/abs/quant-ph/0102040}{arXiv:quant-ph/0102040}.

\bibitem{TranEtAl}
M.~C. Tran, M.~Zuppardo, A.~de Rosier, L.~Knips, W.~Laskowski,
T.~Paterek, and H.~Weinfurter,
\emph{Genuine \(N\)-partite entanglement without \(N\)-partite
correlation functions},
Phys.\ Rev.\ A \textbf{95}, 062331 (2017);
\href{https://arxiv.org/abs/1704.03385}{arXiv:1704.03385}.

\bibitem{BrydgesEtAl}
T.~Brydges, A.~Elben, P.~Jurcevic, B.~Vermersch, C.~Maier,
B.~P. Lanyon, P.~Zoller, R.~Blatt, and C.~F. Roos,
\emph{Probing R\'enyi entanglement entropy via randomized measurements},
Science \textbf{364}, 260--263 (2019);
\href{https://arxiv.org/abs/1806.05747}{arXiv:1806.05747}.

\bibitem{ChristandlSchuchWinter}
M.~Christandl, N.~Schuch, and A.~Winter,
\emph{Entanglement of the antisymmetric state},
Commun.\ Math.\ Phys.\ \textbf{311}, 397--422 (2012);
\href{https://arxiv.org/abs/0910.4151}{arXiv:0910.4151}.

\end{thebibliography}
\endgroup
\end{document}